\documentclass[12pt]{article}

 \usepackage[T1]{fontenc}              
 \usepackage[width=16cm, height=22cm]{geometry}

\usepackage{amsmath,amssymb,bm,color,mathrsfs,mathtools}
 \usepackage[amsmath,thmmarks,hyperref]{ntheorem} 
\usepackage{amscd,verbatim}
\usepackage{bbding,wasysym,pifont}
\usepackage{graphicx}
\usepackage{subfigure}
\usepackage{multirow}
\usepackage{comment}
\usepackage[arrow,matrix,curve]{xy}
\usepackage{subfigure}
\usepackage{authblk}
\usepackage{tikz-cd}
\usepackage{colortbl}
\usepackage{rotating}

\usepackage{tikz-cd}
\usetikzlibrary{decorations.markings}
\tikzset{->-/.style={decoration={
  markings,
  mark=at position #1 with {\arrow{>}}},postaction={decorate}}}
\tikzset{->-/.default=0.5}

\usepackage{pgfplots}
\pgfmathdeclarefunction{gauss}{2}{%
  \pgfmathparse{1/(#2*sqrt(2*pi))*exp(-((x-#1)^2)/(2*#2^2))}%
}
\pgfmathdeclarefunction{nongauss}{1}{%
  \pgfmathparse{1/(1+(x-#1)^2)^(0.5)}%
}

\newcommand\wh{\widehat}

\newcommand{\CC}{{\mathbb C}}
\newcommand{\NN}{{\mathbb N}}
\newcommand{\RR}{{\mathbb R}}
\newcommand{\TT}{{\mathbb T}}
\newcommand{\ZZ}{{\mathbb Z}}

\newcommand{\cF}{{\mathcal F}}
\newcommand{\cK}{{\mathcal K}}
\newcommand{\cH}{{\mathcal H}}
\newcommand{\cL}{{\mathcal L}}

\newcommand{\cS}{{\mathcal S}}

\newcommand{\cE}{\mathcal{E}}

\newcommand{\dd}{\mathrm{d}}

\newcommand{\pg}{\mathsf{pg}}

\theoremstyle{nonumberplain}  
\theoremheaderfont{\itshape}  
\theorembodyfont{\normalfont}  
\theoremseparator{.}  
\theoremsymbol{$\Box$}
\newtheorem{proof}{Proof} 
  
\theoremstyle{plain}  
\theoremheaderfont{\normalfont\bfseries}  
\theorembodyfont{\itshape}  
\theoremsymbol{~}
\theoremseparator{.}  
\newtheorem{proposition}{Proposition}[section]  
  
\newtheorem{corollary}[proposition]{Corollary}  
\newtheorem{lemma}[proposition]{Lemma}  
\newtheorem{theorem}[proposition]{Theorem}   

\theorembodyfont{\normalfont}  
\newtheorem{problem}{Problem} 
\newtheorem{remark}[proposition]{Remark}
\newtheorem{example}[proposition]{Example}

\newtheorem{definition}[proposition]{Definition}

\theoremstyle{nonumberplain}
\theoremheaderfont{\normalfont\bfseries}  
\theorembodyfont{\itshape}  
\theoremsymbol{~}
\theoremseparator{.}

\begin{document}

\title{Good Wannier bases in Hilbert modules associated to topological insulators}

\author[1]{Matthias Ludewig}

\author[2]{Guo Chuan Thiang}
\affil[1,2]{School of Mathematical Sciences, University of Adelaide, SA 5005, Australia}


\maketitle

\begin{abstract}
For a large class of physically relevant operators on a manifold with discrete group action, we prove general results on the (non-)existence of a basis of well-localised Wannier functions for their spectral subspaces. This turns out to be equivalent to the freeness of a certain Hilbert module over the group $C^*$-algebra canonically associated to the spectral subspace. This brings into play $K$-theoretic methods and justifies their importance as invariants of topological insulators in physics.
\end{abstract}

\setcounter{tocdepth}{1}

\section{Introduction} \label{SectionIntroduction}
In solid state physics, one often studies the Schr\"{o}dinger equation on $L^2(\RR^d)$ with a potential which is periodic under the action of a lattice $\ZZ^d$ of translations preserving an underlying crystalline structure (the atomic positions, say). The spectral subspace $L^2_S(\RR^d)$ corresponding to the spectrum of the Hamiltonian operator lying between some spectral gaps is then invariant under the $\ZZ^d$ translations. A (composite) \emph{Wannier basis} for $L^2_S(\RR^d)$ comprises a set of wavefunctions $w_j, j=1,\ldots, n$ and their translates $\gamma^*w_j$ by $\gamma\in\ZZ^d$, such that the set
\begin{equation*}
\gamma^*w_j, \qquad \text{with} \qquad j=1,\dots,n, \quad \gamma\in\ZZ^d
\end{equation*}
 is an orthonormal basis for $L^2_S(\RR^d)$, thus identifying $L^2_S(\RR^d)$ as $n$ copies of the regular representation of $\ZZ^d$ sitting inside $L^2(\RR^d)$. Wannier bases are convenient for expanding the ``effective'' electron states in a spectral subspace of physical interest, and it is often desirable to choose the $w_j$ to be as localised as possible, so that one may reasonably think of the $\gamma^*w_j$ as ``atomic orbitals'' localised at the atomic positions labelled by $\gamma\in\ZZ^d$, see Fig.~\ref{fig:Zgoodwannier} for an illustration.

Wannier basis construction usually proceeds via the Bloch--Floquet decomposition over the character space $\TT^d$ of the translation symmetry group $\ZZ^d$ (reviewed in \S\ref{sec:BlochFloquetreview}). From this vantage point, much effort has been devoted to proving the existence, for arbitrary $L^2_S(\RR^d)$, of ``good'' Wannier bases with, say, exponential decay \cite{Kohn,DesCloi,Nenciu,Brouder,Kuchment,Monaco,Cornean,CMM}. Remarkably, in $d\geq 2$, there is a \emph{topological} obstruction --- the Chern class of the so-called Bloch bundle over $\TT^2$ \cite{Brouder} --- which persists even if we relax the the decay condition significantly (cf.\ Remark \ref{RemarkDecay}). The non-existence of good Wannier bases, or ``atomic limits'' \cite{Bradlyn}, is a paradigmatic feature of so-called \emph{topological insulators}. When a boundary is subsequently introduced, the slow decay rate of bad Wannier bases means that there is no meaningful length scale to decouple the ``bulk'' degrees of freedom from the ``boundary'' ones. This is one way to see why a bulk-boundary correspondence \cite{BCR,EwertMeyer,Kellendonk,Kubota,Ludewig-Thiang,HMT,PSB} should be expected in topological insulators.

\medskip

How much of these insights depend on having an abelian symmetry group and the classical Fourier transform? In this paper, we abstract the salient features of the ``good Wannier basis existence problem'', and demonstrate that the above topology/localization dichotomy \cite{Brouder,Kuchment,Monaco,CMM} holds much more generally, for nonabelian symmetry groups, and also for projective symmetry group representations (whether the group is abelian or not) which occur for quantum Hall Hamiltonians (see \S \ref{sec:twisted.algebras}).

\medskip

\noindent{\bf Basic Setup.} Let $X$ be a complete connected Riemannian manifold with an effective, cocompact, properly discontinuous, isometric action of a countable group $\Gamma$. Let $E$ be an complex $\Gamma$-equivariant vector bundle over $X$ with a $\Gamma$-invariant fiber metric.

\medskip

The relevant space of fields is then the space $L^2(X, E)$ of square-integrable sections\footnote{If $E$ is just a trivialised line bundle, we will simply write $L^2(X)$, omitting reference to $E$.} of $E$, which admits a right action of the reduced group $C^*$-algebra $C^*_r(\Gamma)$ (see \S\ref{SectionRapidDecay}) in a canonical way. Now given a Hamiltonian $D$, i.e.\ a self-adjoint operator acting on sections of $L^2(X, E)$, we assume that we are given a compact subset $S \subset \mathrm{Spec}(D)$ of the spectrum of $D$, which is separated from the rest of the spectrum by spectral gaps. A typical example has the spectrum of $D$ bounded from below, and the so-called \emph{Fermi level} lying in a spectral gap ---  this describes an \emph{insulator}. The spectral subspace $L^2_S(X,E)$ for energies below the Fermi level is the subspace of occupied energy states, and is of particular physical interest in determining various material properties.

The idea is now to consider the spectral subspace $L^2_S(X, E)$ as a module over the reduced group $C^*$-algebra $C^*_r(\Gamma)$. However, in general, this space is too large to be finitely generated and projective. In the case $\Gamma = \ZZ^d$, where $C^*_r(\Gamma) \cong C(\TT^d)$, continuous functions on the Brillouin torus, this parallels the fact that $L^2_S(\RR^d)$ is identified with the space of square-integrable (not necessarily continuous) sections of the Bloch bundle; clearly this is not finitely generated as a module over $C(\TT^d)$ and moreover does not carry any topological information.

To remedy the situation, we use a construction of Roe \cite[pp.~243]{Roe} which identifies a dense subspace $L^2_\Gamma(X, E) \subset L^2(X, E)$, which is a Hilbert module over $C^*_r(\Gamma)$, in other words, admits a $C^*_r(\Gamma)$-valued scalar product. The construction is such that in the abelian case $\Gamma = \ZZ^d$, the intersection $L^2_S(\RR^d) \cap L^2_\Gamma(\RR^d)$ is precisely the space of {\em continuous} sections of the Bloch bundle, which is finitely generated and projective as a $C(\TT^d)$-module, by the theorem of Serre--Swan.

We then prove the following result.

\begin{theorem} \label{MainTheoremIntro}
Assume in addition that $\Gamma$ has \emph{polynomial growth}. Let $D$ be a self-adjoint equivariant differential operator acting on sections of $E$, which is either of Laplace type, or first order elliptic. Suppose that $S$ is a compact subset of the spectrum which is separated from the rest of the spectrum and let $L^2_S(X, E)$ be the corresponding spectral subspace. Then the subspace 
\begin{equation*}
P_S := L^2_S(X, E) \cap L^2_\Gamma(X, E)
\end{equation*}
 is a finitely generated, projective $C^*_r(\Gamma)$-module. Moreover, the following assertions are equivalent.
\begin{enumerate}
\item[{\normalfont ($i$)}] $P_S$ is a free $C^*_r(\Gamma)$-module of rank $n$.
\item[{\normalfont ($ii$)}] There exist functions $w_1, \dots, w_n \in \cS(X, E)$ such that the set $\gamma^* w_j$, $j=1, \dots, n$, $\gamma \in \Gamma$ is an orthonormal basis of $L^2_S(X, E)$.
\end{enumerate}
Here $\cS(X, E)$ denotes the space of smooth sections of $E$ which decay faster than any polynomial, together with their derivatives, c.f.\ Def.~\ref{defn:Schwartzclass} below.
\end{theorem}

In fact, we also prove a stronger, quantitative version of the theorem above, generalizing the results of Kuchment \cite{Kuchment} (see also \cite{CMM}): Namely, we show that for any $n \in \NN$, the existence of a module $Q$ such that $P_S \oplus Q$ is free of rank $n$, is equivalent to the existence of a {\em tight frame} of $L^2_S(X,E)$ consisting of $n$ ``good'' Wannier functions $w_1, \dots, w_n\in \cS(X,E)$ and their $\Gamma$-translates; c.f.~\S\ref{main.theorem.section}.

\medskip

The above results do not require $\Gamma$ to be abelian, so they apply, for instance, to the theory of \emph{crystalline} topological insulators in solid state physics, for which $X = \RR^d$, the Euclidean space, and $\Gamma$ is a \emph{generally non-abelian} crystallographic group (with $E$ trivial). Any such group has polynomial growth; we remark that by Gromov's theorem \cite{Gromov}, groups of polynomial growth are precisely those having a nilpotent subgroup of finite index. Our results show that the $K$-theory of $C^*_r(\Gamma)$, which is a proposed way to classify crystalline topological insulators \cite{FM,Thiang}, presents \emph{computable} obstructions to the existence of good Wannier bases (in the sense of being in the Schwartz class). Furthermore, these topological insulator inspired ideas extend to the \emph{non-Euclidean} setting, as we exhibit in an example in \S\ref{sec:non-Euclidean.example}. Let us mention that $C^*$-algebra methods have previously been applied to the mathematics of topological insulators in (implicitly) \emph{Euclidean} settings, see e.g.\ \cite{Bellissard,BCR,EwertMeyer,Kellendonk,Kubota,PSB}, and independently of the Wannier basis construction problem.

\medskip

We remark that we could also directly consider $L^2_S(X, E)$ as a module over the group von-Neumann-algebra $\mathcal{N}(\Gamma)$, which may seem simpler than considering the intersection $L^2_S(X, E) \cap L^2_\Gamma(X, E)$ necessary in the $C^*_r(\Gamma)$-case. However, we lose topological information: For example, in the case $\Gamma = \ZZ^d$, we will always have $L^2_S(X, E) \cong \ell^2(\ZZ^d)^n$ as $\mathcal{N}(\ZZ^d)$-modules, corresponding to the fact that the Bloch bundle aways has a square-integrable (non-continuous) trivialization.

\medskip

After fleshing out the constructions and proof of the main theorem in \S\ref{SectionRapidDecay}-\S\ref{main.theorem.section}, we formulate the general ``good Wannier basis existence problem'' (Problem \ref{prob:good.Wannier}) in the final section \S\ref{sec:physics}, and apply our theory to compute explicitly the presence/absence of obstructions in several physically interesting new examples. For the convenience of readers who are not acquainted with noncommutative geometry ideas, we include Table \ref{table:dictionary} to convert some familiar ideas in usual Bloch theory to the setup of non-commutative geometry.

\begin{table}
\begin{center}
\begin{tabular}{ m{6cm} || m{6cm}}
 \hline
 \hline
 $\ZZ^d$ acting on Euclidean $\RR^d$    & Non-abelian $\Gamma$ acting on Riemannian manifold $X$ \\
 \hline
 \hline
 Unit cell $\RR^d/\ZZ^d$   & Fundamental domain $X/\Gamma$    \\
  \hline
Continuous functions on quasi-momentum space / Brillouin torus $\TT^d$ &  Group $C^*$-algebra $C^*_r(\Gamma)$  \\
 \hline
Bloch--Floquet transform & $\Phi_\mathcal{F}$, Eq.\ \eqref{IsoPhi}\\
\hline
Continuous sections of Bloch bundle over $\TT^d$  & Hilbert $C^*_r(\Gamma)$-module $L^2_\Gamma(X)$ inside $L^2(X)$ \\
\hline
Hamiltonian with Band structure & $C^*_r(\Gamma)$-elliptic operator on $L^2_\Gamma(X)$ \\
\hline
Smooth functions on $\TT^d$ & The subalgebra $H^\infty(\Gamma) \subset C^*_r(\Gamma)$ \\
\hline
Finite-rank spectral subbundle of Bloch bundle (``spectrally isolated bands'') & $P_S=L^2_S(X)\cap L^2_\Gamma(X)$ as a finitely generated projective submodule of $L^2_\Gamma(X)$  \\
\hline
Topologically (resp.\ smoothly) trivialialisable bands & Free Hilbert $C^*_r(\Gamma)$-module (resp.\ pre-Hilbert $H^\infty(\Gamma)$-module)  \\
\hline
Chern/topological $K$-theory class of bands & Noncommutative Chern class/ operator $K$-theory class of $P_S$\\
\hline
Orthonormal Wannier basis functions $\{w_j\}_{j=1,\ldots,n}$/``atomic limit''& Free (pre)-Hilbert module generators\\
 \hline
\end{tabular}
\caption{A dictionary to convert between the usual commutative Bloch theory constructions, into noncommutative geometry language. More precise descriptions of the right-side can be found in the main text.}\label{table:dictionary}
\end{center}
\end{table}

\section{Algebras of rapidly decaying sequences} \label{SectionRapidDecay}

Let $\Gamma$ be a finitely generated group. Its {\em group algebra} $\CC[\Gamma]$ is the set of all finite formal linear combinations of the group elements $\gamma\in\Gamma$. It has a $*$-involution given by
\begin{equation}
  \Bigl(\sum_{\gamma \in \Gamma} a_\gamma \gamma\Bigr)^* = \sum_{\gamma \in \Gamma} \overline{a}_\gamma \gamma^{-1}.\label{eqn:group.algebra.involution}
\end{equation}
The group algebra $\CC[\Gamma]$ acts on the Hilbert space
\begin{equation*}
  \ell^2(\Gamma) = \Bigl\{ b = \sum_{\gamma \in \Gamma} b_\gamma \gamma ~\Bigl|~ \|b\|_{\ell^2(\Gamma)} := \sum_{\gamma \in \Gamma} |b_\gamma|^2 < \infty \Bigr\}
\end{equation*}
by translation (i.e.\ left regular representation),
\begin{equation*}
(a\cdot b)_\gamma=\sum_{\rho\in\Gamma} a_\rho b_{\gamma\rho^{-1}}.
\end{equation*}
 Multiplication by $a \in \CC[\Gamma]$ is bounded, hence one obtains a representation of $\CC[\Gamma]$ on the bounded operators $\cL(\ell^2(\Gamma))$, which is in fact a $*$-representation. The {\em reduced group $C^*$-algebra} is then the $C^*$-algebra $C^*_r(\Gamma) \subseteq \cL(\ell^2(\Gamma))$ obtained by completing $\CC[\Gamma]  \subseteq \cL(\ell^2(\Gamma))$ with respect to the operator norm. Denoting by $u \in \ell^2(\Gamma)$  the element which is one at the unit of $\Gamma$ and zero otherwise, the map
\begin{equation*}
  C^*_r(\Gamma) \longrightarrow \ell^2(\Gamma), \quad a \longmapsto a \cdot u
\end{equation*}
provides a continuous embedding of $C^*_r(\Gamma)$ into $\ell^2(\Gamma)$, and further into the space $\ell^\infty(\Gamma)$ of all bounded, $\Gamma$-indexed sequences. In particular, elements of $C^*_r(\Gamma)$ are bounded $\Gamma$-indexed sequences.

Let $\mathscr{S}$ be a finite generating set of $\Gamma$ with $1\not\in \mathscr{S}$ and $\mathscr{S}^{-1} = \mathscr{S}$, and let 
\begin{equation*}
  L(\gamma) = \min \bigl\{ n \in \NN_0 \mid \exists s_1,\dots, s_n \in \mathscr{S}:  \gamma = s_1\cdots s_n \bigr\}
\end{equation*}
be the corresponding {\em length function}. Now for any auxiliary Banach space $A$, we define
\begin{equation}
  H^\infty(\Gamma, A) := \Bigl\{ a=\sum_{\gamma \in \Gamma} a_\gamma \gamma  ~\Bigl|~ \sum_{\gamma \in \Gamma} |a_\gamma|_A^2 L(\gamma)^s < \infty~~\text{for all}~~ s \geq 0.\Bigr\},\label{eqn:rapid.decrease.sequence.definition}
\end{equation}
the space of {\em rapidly decreasing sequences} indexed by $\Gamma$ with values in $A$. 

\begin{remark}
Above, we wrote $|a_\gamma|_A$ for the pointwise norms of the $A$-valued sequence $a\in H^\infty(\Gamma, A)$, and we will also write $|a|_A\in H^\infty(\Gamma,\CC)$ for the sequence of pointwise $A$-norms. When $A=\CC$, we simply write $H^\infty(\Gamma)=H^\infty(\Gamma, \CC)$, which can be thought of as the space of as ``smooth functions on the dual of $\Gamma$''. For example, $H^\infty(\ZZ)$ comprises the Fourier coefficient sequences of smooth functions on the dual circle $\TT=\wh{\ZZ}$.
\end{remark}

The space $H^\infty(\Gamma, A)$ is topologized by the increasing sequence of Hilbert space norms
\begin{equation} \label{SobolevNorm}
  \|a\|_{s, A}^2 = \sum_{\gamma \in \Gamma} |a_\gamma|_A^2 \bigl(1+L(\gamma)\bigr)^{2s},
\end{equation}
which turn $H^\infty(\Gamma, A)$ into a Fr\'echet space. One easily checks that the definition of $H^\infty(\Gamma, A)$ is independent of the choice of finite generating set $\mathscr{S}$; in particular, any such choice gives rise to an equivalent set of norms \eqref{SobolevNorm}.

\begin{lemma} \label{LemmaPolGrowth}
Assume that $\Gamma$ has {\em polynomial growth}, meaning that
\begin{equation} \label{PolyGrowth}
  \sum_{\gamma \in \Gamma} \bigl(1+L(\gamma)\bigr)^{-s} < \infty
\end{equation}
for some $s >0$ large enough. Then for any $C^*$-algebra $A$, the space $H^\infty(\Gamma, A)$ is a Fr\'echet algebra in a natural way, which is continuously included in the tensor product\footnote{For two $C^*$-algebras $A_1\subset \cL(\cH_1)$ and $A_2\subset \cL(\cH_2)$, their algebraic tensor product sits inside \mbox{$\cL(\cH_1\otimes\cH_2)$} and its completion therein is the spatial tensor product, see Appendix T.5 of \cite{WO} for details of this and other possible $C^*$-algebra tensor products. In our case where $\Gamma$ has polynomial growth and is hence amenable, its group $C^*$-algebra is nuclear. So for any $C^*$-algebra $A$, there is actually a unique $C^*$-algebra tensor product $A \otimes C^*_r(\Gamma)$. For details, see \cite{BrownOzawa} Example~2.6.6, Thm.~2.6.8 and Thm.~3.8.7.}
$A \otimes C^*_r(\Gamma)$.
\end{lemma}

\begin{proof}
We first show that there exist $C, s > 0$ such that
\begin{equation} \label{NormEstimateS}
  \|a\|_{A \otimes C^*_r(\Gamma)} \leq C \|a\|_{s, A}
\end{equation}
for all $a \in H^\infty(\Gamma, A)$, which implies that $H^\infty(\Gamma, A)$ is continuously included in $A \otimes C^*_r(\Gamma)$.
Namely, for $a \in A \otimes \CC[\Gamma]$ (the algebraic tensor product), we have by the triangle inequality that
\begin{equation*}
  \|a\|_{A\otimes C^*_r(\Gamma)} \leq \sum_{\gamma \in \Gamma} \|a_\gamma \otimes \gamma\|_{A\otimes C^*_r(\Gamma)} = \sum_{\gamma \in \Gamma} \|a_\gamma\|_A.
\end{equation*}
ence by the Cauchy--Schwarz inequality,
\begin{equation*}
\begin{aligned}
  \|a\|_{A\otimes C^*_r(\Gamma)} &\leq \left(\sum_{\gamma \in \Gamma} |a_\gamma|_A^2\bigl(1+L(\gamma)\bigr)^{2s}\right)^{1/2} \left(\sum_{\gamma \in \Gamma} \bigl(1+L(\gamma)\bigr)^{-2s}\right)^{1/2}
  \end{aligned}
\end{equation*}
with the second $a$-independent factor finite by \eqref{PolyGrowth}. Since $A \otimes \CC[\Gamma]$ is dense in $A \otimes C^*_r(\Gamma)$ and in $H^\infty(\Gamma, A)$, this shows the claim.

Clearly, $H^\infty(\Gamma, A)$ is an algebra in the obvious way. To show that the multiplication is continuous, we will show that for any $s \geq 0$, there exist $C^\prime, r > 0$ such that
\begin{equation} \label{ProductNormEstimateS}
  \|a \cdot b\|_{s, A} \leq C^\prime \|a\|_{r, A} \|b\|_{r, A}
\end{equation}
for all $a, b \in H^\infty(\Gamma, A)$. It suffices to consider the case $s=n \in \NN$. Then
\begin{equation*}
  \|a\cdot b\|_{n, A}^2 = \sum_{\gamma \in \Gamma} \Bigl| \sum_{\rho \in \Gamma} a_\rho b_{\rho^{-1} \gamma} \Bigr|^2 L(\gamma)^{2n} 
  \leq \sum_{\gamma \in \Gamma} \Bigl| \sum_{\rho \in \Gamma} |a_\rho|_A |b_{\rho^{-1} \gamma}|_A L(\gamma)^n \Bigr|^2,
\end{equation*}
Using the {\em triangle inequality}
\begin{equation} \label{TriangleInequality}
    L(\gamma\rho) \leq L(\gamma) + L(\rho) \quad \text{for all} \quad \gamma, \rho \in \Gamma,
\end{equation}
we obtain
\begin{equation*}
\begin{aligned}
  \|a\cdot b\|_{n, A} &\leq \sum_{k=0}^n \binom{n}{k} \bigl\| L^k|a|_A \cdot L^{n-k}|b|_A\bigr\|_{\ell^2(\Gamma)} \\
  &\leq \sum_{k=0}^n \binom{n}{k} \bigl\| L^k|a|_A\bigr\|_{\mathrm{op}} \bigl\|L^{n-k}|b|_A\bigr\|_{\ell^2(\Gamma)} \\
  &\leq C \sum_{k=0}^n \binom{n}{k} \bigl\| L^k|a|_A\bigr\|_{s} \bigl\||b|_A\bigr\|_{n-k}\\
  &= C \sum_{k=0}^n \binom{n}{k} \| a\bigr\|_{s+k, A} \bigl\|b\bigr\|_{n-k, A},
\end{aligned}
\end{equation*}
where $|a|_A, |b|_A \in H^\infty(\Gamma)$ are the respective pointwise $A$-norms of $a, b$; moreover, we used the estimate \eqref{NormEstimateS}. Since the norms \eqref{SobolevNorm} are increasing in strength, this implies \eqref{ProductNormEstimateS} with $r = s+n$. 
\end{proof}

\begin{remark} \label{RemarkPropertyRD}
More generally, a group $\Gamma$ is said to have property (RD), if there exist $C, s>0$ such that $\|a\|_{\mathrm{op}} \leq C \|a\|_s$ for all $a \in \CC[\Gamma]$, which implies $H^\infty(\Gamma) \subseteq C^*_r(\Gamma)$. The lemma above shows that groups of polynomial growth have property (RD). Conversely, it is known that an amenable group has property (RD) if and only if it is of polynomial growth \cite{Jolissaint}. The proof above also shows that also under the assumption (RD), $H^\infty(\Gamma, A)$ is a Fr\'echet algebra for any $C^*$-algebra $A$, i.e.\ that we have estimates of the form \eqref{ProductNormEstimateS}.
\end{remark}

A subalgebra $S$ of a $C^*$-algebra $B$ is called {\em spectral}, if for any $b \in S$ that is invertible in $B$, we already have $b^{-1} \in S$. Moreover, we say that $S$ is {\em closed under holomorphic functional calculus} when $f(b) \in S$ for any function $f$ that is is holomorphic on an open neighborhood of the spectrum of $b$ (in $B$). 

\begin{proposition} \label{PropSpectral}
  Suppose that $\Gamma$ has polynomial growth. Then for any unital $C^*$-algebra $A$, the dense subalgebra $H^\infty(\Gamma, A)$ of $A \otimes C^*_r(\Gamma)$ is spectral and closed under holomorphic functional calculus.
\end{proposition}

In order to prove this proposition, we use the following criterion of Ji, Thm.~1.2 of \cite{Ji}. In fact, Ji proves a version of Prop.~\ref{PropSpectral} for the special case where $A = \mathrm{M}_n(\CC)$, a finite-dimensional algebra, but under the weaker condition that $\Gamma$ only satisfies property (RD). We adapt his proof to show that under the stronger condition of polynomial growth, one can allow $A$ to be arbitrary.

\begin{proposition} \label{PropJi}
Let $L$ be a unital $C^*$-algebra and $B \subseteq L$ be a closed subalgebra containing the unit of $L$. Let
\begin{equation*}
  \delta: L \supset \mathrm{dom}(\delta) \longrightarrow L
\end{equation*}
be a closed, unbounded derivation. Then $S := \bigcap_{n=1}^\infty \mathrm{dom}(\delta^n) \cap B$ is a subalgebra of $B$, and if $S$ is dense in $B$, then it is spectral in $B$ and closed under holomorphic functional calculus.
\end{proposition}

We remark that the theorem of Ji does not include the additional statement that $S$ is closed under holomorphic functional calculus. However, this follows easily from an additional argument by Valette, see \cite{Valette}, below Prop.~8.12.

\begin{proof}[of Prop.~\ref{PropSpectral}]
Let $A \subset \cL(H)$ for some Hilbert space $H$. We will use Prop.~\ref{PropJi} for $L = \cL(H \otimes \ell^2(\Gamma))$ and $B = A \otimes C^*_r(\Gamma) \subseteq L$. To get our hands on a derivation $\delta$, we define the densely defined, unbounded operator 
\begin{equation*}
  M_L: H \otimes \ell^2(\Gamma) \supset \mathrm{dom}(M_L)\longrightarrow H \otimes \ell^2(\Gamma)
\end{equation*}
given by pointwise multiplication with the length function $L$, where we set $\mathrm{dom}(M_L) := H \otimes \CC[\Gamma]$, making $M_L$ densely defined. Now define
\begin{equation*}
\mathrm{dom}(\delta) = \bigl\{ a \in A \otimes C^*_r(\Gamma) \mid [M_L, a] ~\text{is densely defined and bounded} \bigr\}.
\end{equation*}
Hence for any $a \in \mathrm{dom}(\delta)$, the operator $[M_L, a]$ extends by continuity to an element of $\cL(H \otimes \ell^2(\Gamma))$, again denoted by $[M_L, a]$; the operator $\delta$ is then given by
\begin{equation*}
  \delta: A \otimes C^*_r(\Gamma) \supset \mathrm{dom}(\delta) \longrightarrow  \cL\bigl(H \otimes \ell^2(\Gamma)\bigr), \quad a \longmapsto [M_L, a].   
\end{equation*}
It is straightforward to show by induction that for any $n \in \NN$, we have the formula 
\begin{equation} \label{FormulaDeltaN}
  \delta^n(a) h = \sum_{\gamma, \rho \in \Gamma} \bigl(L(\gamma) - L(\rho^{-1}\gamma)\bigr)^n \cdot a_\rho h_{\rho^{-1}\gamma}  \gamma
\end{equation}
for $a \in A \otimes C^*_r(\Gamma)$ and $h \in H \otimes \CC[\Gamma]$ (note that this is a well-defined, not necessarily square-summable, $\Gamma$-indexed sequence, as for fixed $\gamma$, the sum over $\rho$ is in fact finite). From the triangle inequality \eqref{TriangleInequality}, we obtain for $a \in A \otimes \CC[\Gamma]$, $h \in H \otimes \CC[\Gamma]$ that
\begin{equation*}
\begin{aligned}
  \|\delta^n(a) h\|_{\ell^2(\Gamma)\otimes H}^2 &\leq \sum_{\gamma\in \Gamma} \Bigl| \sum_{\rho \in \Gamma} L(\rho)^n|a_\rho|_A |h_{\rho^{-1}\gamma}|_H\Bigr|^2 = \bigl\| L^n|a|_A \cdot |h|_H\bigr\|_{\ell^2(\Gamma)}^2 \\
  &\leq \bigl\| L^n|a|_A \bigr\|_{\mathrm{op}}^2\bigl\| |h|_H\bigr\|_{\ell^2(\Gamma)}^2 = \bigl\| L^n|a|_A \bigr\|_{\mathrm{op}}^2\|h\|_{H \otimes \ell^2(\Gamma)}^2.
\end{aligned}
\end{equation*}
Here $|a|_A$ denotes the element of $C^*_r(\Gamma)$ (as a bounded $\Gamma$-indexed sequence) obtained by taking the pointwise $A$-norm of $a$; similarly for $|h|_H \in \ell^2(\Gamma)$. This shows that $A \otimes \CC[\Gamma] \subset \mathrm{dom}(\delta^n)$ for each $n \in \NN$, hence the intersection of all $\mathrm{dom}(\delta^n)$ is dense. Moreover, since $\Gamma$ has polynomial growth, Lemma~\ref{LemmaPolGrowth} implies that there exist $C, s \geq 0$ such that
\begin{equation*}
\bigl\| L^n|a|_A \bigr\|_{\mathrm{op}} \leq C \|L^n|a|_A\|_s = C \|a\|_{s+n, A}.
\end{equation*} 
This shows that
\begin{equation*}
  \|\delta^n(a)\|_{A \otimes C^*_r(\Gamma)} \leq C \|a\|_{s+n, A},
\end{equation*}
which implies that $H^\infty(\Gamma, A) \subseteq \mathrm{dom}(\delta^n)$, for any $n \in \NN$. 

To see that $\delta$ is closed, let $a_n \in \mathrm{dom}(\delta)$ converge to some $a \in A \otimes C^*_r(\Gamma)$. We claim that this means in particular that $a_{n, \gamma} \rightarrow a_\gamma$ in $A$ for every $\gamma \in \Gamma$. Namely, given $\xi \in H$, denote by $u_\xi$ the element of $\ell^2(\Gamma) \otimes H$ which is $\xi$ at the unit element of $\Gamma$ and zero otherwise. Then we have
\begin{equation*}
  |a_{n, \gamma} - a_\gamma|^2_A \leq \sup_{|\xi|=1}\sum_{\gamma \in \Gamma} \bigl|(a_{n, \gamma} - a_\gamma)\xi\bigr|_H^2 = \sup_{|\xi|=1}\bigl\|(a_n-a)u_\xi\bigr\|_{H \otimes \ell^2(\Gamma)}^2 \leq \|a_n-a\|^2_{A \otimes C^*_r(\Gamma)},
\end{equation*}
which converges to zero as $n \rightarrow \infty$, proving the claim. Having established this fact, formula \eqref{FormulaDeltaN} implies that for all $h \in H \otimes \CC[\Gamma]$,
\begin{equation*}
\delta(a_n)h \longrightarrow [M_L, a]h
\end{equation*}
pointwise, in the sense that the coefficients $(\delta(a_n)h)_\gamma$ converge in $H$ for each $\gamma \in \Gamma$; here $[M_L, a]h$ is some $\Gamma$-indexed sequence, not necessarily square-summable. But now if $\delta(a_n)$ converges to $b \in \cL(H \otimes \ell^2(\Gamma))$, then in particular, $\delta(a_n)h$ converges pointwise to $bh$, for all $h \in H \otimes \CC[\Gamma]$. Hence in fact $bh = [M_L, a]h$ for each $H \otimes \CC[\Gamma]$. We obtain that $H \otimes \CC[\Gamma] \subset \mathrm{dom}([M_L, a])$, hence $[M_L, a]$ is densely defined, and moreover, $[M_L, a]$ is bounded, since $b$ is bounded. This shows $a \in \mathrm{dom}(\delta)$, and $b = \delta(a)$, hence $\delta$ is a closed operator.

With a view on the result of Ji cited above, it remains to show that 
\begin{equation*}
  \bigcap_{n=1}^\infty \mathrm{dom}(\delta^n) \subseteq H^\infty(\Gamma, A)
\end{equation*}
in order to show that $H^\infty(\Gamma, A)$ is spectral in $A \otimes C^*_r(\Gamma)$, and \emph{it is this step where we need the stronger condition of polynomial growth instead of just property} (RD).
Namely, let $a \in \mathrm{dom}(\delta^n)$ for any $n \in \NN$. Given $\xi \in H$, denote by $u_\xi$ the element of $\ell^2(\Gamma) \otimes H$ which is $\xi$ at the unit element of $\Gamma$ and zero otherwise. Then as $a \in \mathrm{dom}(\delta^n)$, we have
\begin{equation*}
\delta^n(a) u_\xi = \sum_{\gamma \in \Gamma} L(\gamma)^n a_\gamma(\xi)  \in \ell^2(\Gamma) \otimes H,
\end{equation*}
 in fact 
\begin{equation*}
 C_n^2 |\xi|_H^2 = C_n^2 \|u_\xi\|_{\ell^2(\Gamma) \otimes H}^2 \geq \|\delta^n(a) u_\xi\|_{\ell^2(\Gamma) \otimes H}^2 = \sum_{\gamma \in \Gamma} L(\gamma)^{2n} |a_\gamma(\xi)|_H^2 
\end{equation*}
for some constant $C_n>0$, depending only on $n$. In particular, we have that $|a_\gamma|_A \leq C_n L(\gamma)^{-n}$, for any $n \in \NN$. Since $\Gamma$ has  polynomial growth, this implies that $a \in H^\infty(\Gamma, A)$.
\end{proof}

\section{The Hilbert module of a group action} \label{SectionExamples}

In this section, we review the notion of (Pre-)Hilbert modules over $C^*$-algebras, and introduce the basic examples relevant to our paper.

Let $A$ be a pre-$C^*$-algebra, by which we mean a $*$-closed subalgebra of a $C^*$-algebra. By a {\em pre-Hilbert $A$-module}, we mean a right $A$-module $M$ together with an $A$-valued inner product,
\begin{equation*}
  ( \, \cdot \, | \, \cdot \,) : M \times M \longrightarrow A,
\end{equation*}
which is $\CC$-antilinear in the first argument, $\CC$-linear in the second argument, and satisfies the identities
\begin{equation}\label{ScalarProductProperties}
  (x | ya) = (x | y)a \quad \text{and} \quad (y | x) = (x | y)^*,
\end{equation}
for all $x, y \in M$ and $a \in A$. Moreover, we require that $(x | x)$ lies in the cone of positive elements of $A$, for each $x \in M$; here we say that an element of $A$ is positive if it is positive in the completion of $A$. Using the $C^*$-property of the norm of $A$, one deduces the Cauchy--Schwarz inequality
\begin{equation} \label{CSU}
  \|(x|y)\|_A \leq \|(x | x)\|_A^{1/2}\|(y | y)\|_A^{1/2},
\end{equation}
which implies that the positive homogeneous functional
\begin{equation} \label{DefinitionNormOnM}
  \|x\|_M := \|(x | x)\|_A^{1/2}
\end{equation}
satisfies the triangle inequality, hence defines a norm on $M$. If $M$ is complete with respect to the norm \eqref{DefinitionNormOnM}, it is called a \emph{Hilbert $C^*$-module over $A$} or simply a {\em Hilbert $A$-module}. Otherwise, starting from a pre-Hilbert $A$-module $M$, one can ``simultaneously complete'' $A$ and $M$ to obtain a Hilbert $C^*$-module, see pp.\ 4--5 of \cite{Lance}.

\begin{example}[Standard Hilbert $A$-module] \label{L2Module}
Let $H$ be a separable Hilbert space and $A$ a $C^*$-algebra. The {\em interior tensor product} 
$H \otimes A$ is defined as follows (for general facts regarding the interior product, which in this case happens to coincide with the exterior product, see e.g.\ \cite[Chapter 4]{Lance}). On the algebraic tensor product of $H$ and $A$, define the inner product
\begin{equation} 
  \Bigl( \sum_\alpha h_\alpha \otimes a_\alpha \Bigl| \sum_\beta k_\beta \otimes b_\beta \Bigr) = \sum_{\alpha, \beta} \langle h_\alpha, k_\beta\rangle_H\,a_\alpha^* b_\beta.\label{eqn:inner.product.definition}
\end{equation}
We then denote by $H \otimes A$ the completion of the algebraic tensor product with respect to the norm given by \eqref{DefinitionNormOnM} in terms of this inner product. For $H = \ell^2(\NN)$, we obtain the {\em standard (countably-generated) Hilbert $A$-module} $\ell^2(A)$ discussed in \cite[p.~237]{WO}.
If we take $H = \CC^n$, we obtain the standard finitely generated free module $A^n$.
\end{example}

\begin{example}
Let $\Gamma$ be a finitely generated group having polynomial growth or, more generally, property (RD), so that the space $H^\infty(\Gamma)$ is a $*$-subalgebra of $C^*_r(\Gamma)$, c.f.\ Prop.~\ref{LemmaPolGrowth} and Remark~\ref{RemarkPropertyRD}.
Given a separable Hilbert space $H$, the space $H^\infty(\Gamma, H)$ as defined in Eq.~\eqref{eqn:rapid.decrease.sequence.definition} is a right $H^\infty(\Gamma)$-module in a natural way. Moreover, with the inner product
\begin{equation} \label{InnerProductOnL2S}
  (h|k) = \sum_{\gamma, \rho \in \Gamma} \langle h_\rho, k_{\rho\gamma}\rangle_H \gamma,
\end{equation}
$H^\infty(\Gamma, H)$ becomes a pre-Hilbert $H^\infty(\Gamma)$-module. This inner product coincides with the one defined in \eqref{eqn:inner.product.definition} above for $A = C^*_r(\Gamma)$. The inner product indeed takes values in $H^\infty(\Gamma)$: For any $s \geq 0$, we have
\begin{equation*}
\begin{aligned}
  \|(h|k)\|_s^2 = \sum_{\gamma \in \Gamma} \left| \sum_{\rho \in \Gamma} \langle h_\rho, k_{\rho\gamma} \rangle_H\right|^2 L(\gamma)^{2s} \leq \sum_{\gamma \in \Gamma} \left| \sum_{\rho \in \Gamma} |h_\rho|_H |k_{\rho\gamma}|_H\right|^{2} L(\gamma)^{2s},
\end{aligned} 
\end{equation*}
hence
\begin{equation}\label{InnerproductEstimate}
  \|(h|k)\|_s \leq \bigl\| |h|_H^* \cdot |k|_H \bigr\|_s \leq C \bigl\||h|_H^*\bigr\|_r \bigl\||k|_H\bigr\|_r = C^\prime \|h\|_{r, H} \|k\|_{r, H},
\end{equation}
using the estimate \eqref{ProductNormEstimateS}. We also have $H^\infty(\Gamma, H) \subseteq H \otimes C^*_r(\Gamma)$, because
\begin{equation*}
  \|(h|h)\|_{\mathrm{op}} \leq C\|(h|h)\|_s \leq C^{\prime\prime} \|h\|_{r, H}^2
\end{equation*}
using estimate \eqref{NormEstimateS} together with \eqref{InnerproductEstimate}. Of course, the inner product \eqref{InnerProductOnL2S} on $H^\infty(\Gamma, H)$ coincides with the restriction of the inner product of $H \otimes C^*_r(\Gamma)$ from Example \ref{L2Module}; this shows that \eqref{InnerProductOnL2S} has all required properties of a $H^\infty(\Gamma)$-valued inner product.
\end{example}

For our last example, we assume the Basic Setup from the introduction; in other words, let $X$ be a complete Riemannian manifold and let $\Gamma$ be a group acting properly discontinuously on $X$ by isometries, such that $X/\Gamma$ is compact. Moreover, let $E$ be a $\Gamma$-equivariant vector bundle on $X$ with a $\Gamma$-equivariant fiber metric. Denote by $L^2(X, E)$ the space of square-integrable sections of $E$. The data induce an isometric right action of $\Gamma$ on $L^2(X, E)$  by pullback, explicitly,
\begin{equation*}
(w \cdot \gamma)(x) := (\gamma^*w)(x) \equiv w\bigl(\gamma \cdot x),
\end{equation*}
for $w \in L^2(X, E)$, $\gamma \in \Gamma$, $x \in X$. With respect to a choice of fundamental domain $\cF$ for the action of $\Gamma$, we obtain a $\Gamma$-equivariant isometry
\begin{equation} \label{IsoPhi}
  \Phi_\cF: L^2(X, E) \longrightarrow L^2(\cF, E) \otimes \ell^2(\Gamma), ~~~~~ w \longmapsto \sum_{\gamma \in \Gamma} \gamma^*w|_\cF \otimes \gamma^{-1},
\end{equation}
where the tensor product is a tensor product of Hilbert spaces; here $\Gamma$ acts on $\ell^2(\Gamma)$ by right multiplication. Notice that we need the action to be effective in order to ensure surjectivity of $\Phi_\cF$. Of course, these right $\Gamma$-actions extend in the obvious fashion to right $*$-module actions of the group algebra $\CC[\Gamma]$.

\begin{example}[The module $L^2_\Gamma(X, E)$]\label{ex:Roe.module}
We now describe a construction due to Roe, cf.\ pp.\ 243 of \cite{Roe} of a Hilbert $C^*_r(\Gamma)$-module $L^2_\Gamma(X, E)$ associated to the action of the group $\Gamma$ on $X$. As a space, $L^2_\Gamma(X, E)$ is a dense subset of $L^2(X, E)$.

On $L^2(X, E)$, we consider the pairing
\begin{equation} \label{eqn:pairing}
  (v | w) = \sum_{\gamma \in \Gamma} \langle \gamma^*v,  w\rangle_{L^2(X, E)} \gamma,
\end{equation}
a priori taking values in the space $\ell^\infty(\Gamma)$ of bounded sequences indexed by $\Gamma$.
Now starting with \emph{compactly supported} $v, w \in L^2_c(X, E)$, we see that $(v | w) \in \CC[\Gamma]$ with $(v | v)$ positive (in $C^*_r(\Gamma)$), and that the pairing satisfies \eqref{ScalarProductProperties}, so $L^2_c(X, E)$ becomes a pre-Hilbert $\CC(\Gamma)$-module (Lemma 2.1 of \cite{Roe}). We then define
\begin{equation} \label{DefL2Gamma}
 L^2_\Gamma(X, E) = ~~\begin{aligned} &\text{completion of}~L^2_c(X, E)~  \text{with respect} \\ &\text{to the norm}~ \|w\|_\Gamma := \|(w|w)\|_{\mathrm{op}}^{1/2}.
 \end{aligned}
\end{equation}
One shows that the right action of $C^*_r(\Gamma)$ preserves $L^2_\Gamma(X, E)$ and that the bracket, restricted to $L^2_\Gamma(X, E)$, takes values in $C^*_r(\Gamma) \subset \ell^\infty(\Gamma)$. This turns $L^2_\Gamma(X, E)$ into a Hilbert $C^*_r(\Gamma)$-module.

We remark that if  $u \in \ell^2(\Gamma)$ is the element which is one at the unit element and zero otherwise, then
\begin{equation*}
  \|(w, w)u\|_{\ell^2(\Gamma)} = \sum_{\gamma \in \Gamma} |\langle w, \gamma^*w\rangle_{L^2(X, E)}|^2 \geq \|w\|_{L^2(X, E)}^2 \|u\|_{\ell^2(\Gamma)},
\end{equation*}
hence $\|w\|_{\Gamma}^2 = \|(w, w)\|_{\mathrm{op}} \geq \|w\|_{L^2(X, E)}^2$ which shows that in fact, the completion \eqref{DefL2Gamma} can be realized as a subspace of $L^2(X, E)$. Therefore, we always stipulate $L^2_\Gamma(X, E) \subseteq L^2(X, E)$.
\end{example}

Given a fundamental domain $\cF$, we see that $\Phi_\cF$ defined in \eqref{IsoPhi} restricts to a vector space isomorphism
\begin{equation*}
  \Phi_\cF : L^2_c(X, E) \longrightarrow L^2(\cF, E) \otimes \CC[\Gamma].
\end{equation*}
Moreover, the bracket \eqref{eqn:pairing} can be equivalently written as
\begin{equation*}
 (v | w) = \sum_{\gamma, \rho} \langle \rho^*v, \gamma^*w\rangle_{L^2(\cF, E)} \rho \gamma^{-1} ,
\end{equation*}
which shows that $\Phi_\cF$ maps the bracket of $L^2_c(X, E)$ to the bracket \eqref{eqn:inner.product.definition}. Extending this observation by continuity, one has the following result.

\begin{proposition} \label{PropIsoToStandardModule}
After choosing a  fundamental domain $\cF$ for the $\Gamma$-action, $\Phi_\cF$ restricts to an isomorphism of Hilbert $C^*_r(\Gamma)$-modules
\begin{equation*}
  \Phi_\cF: L^2_\Gamma(X, E) \longrightarrow L^2(\cF, E) \otimes C^*_r(\Gamma),
\end{equation*}
where on the right hand side, we have the Hilbert $C^*_r(\Gamma)$-module from Example~\ref{L2Module}.
\end{proposition}

\begin{remark} \label{RemarkDecay}
It is easy to see that if $w \in L^2(X, E)$ has the decay condition
\begin{equation}
  \sum_{\gamma \in \Gamma} \|\gamma^*w\|_{L^2(\cF, E)} < \infty\label{eqn:Kuchment.decay}
\end{equation}
with respect to some fundamental domain $\cF$, then $w \in L^2_\Gamma(X, E)$. Condition Eq.~\eqref{eqn:Kuchment.decay} was considered in \cite{Kuchment}, for $X=\RR^d, \Gamma=\ZZ^d$, and trivial $E$, in connection with Chern class obstructions to localised Wannier bases (see \S \ref{sec:physics} for more on Wannier bases). Also in this abelian setting with $d\leq 3$, another decay condition
\begin{equation}
\sum_{j=1}^n\int_{\RR^d}|x|^{2s}|w_j(x)|^2\;<\infty\label{eqn:power.decay}
\end{equation}
was considered in \cite{Monaco} for the related problem of finding the \emph{optimal} decay of a (composite) Wannier basis for a possibly topologically non-trivial Bloch bundle. Amongst other results, it was found that any $s<1$ could always be achieved whether or not the Bloch bundle is topologically trivial; the threshold case $s=1$, corresponding to $H^1$-trivialisability of the Bloch bundle, is attainable by some Wannier basis exactly when the Bloch bundle is topologically trivialisable. Unfortunately, it seems that there is no simple relationship between the decay condition \eqref{eqn:power.decay} and membership in our $L^2_{\ZZ^d}(\RR^d,E)$.
\end{remark}

\subsection{Comparison with literature}\label{rem:Gruber.comparison}
Roe constructed the Hilbert $C^*_r(\Gamma)$-module $L^2_\Gamma(X,E)$ of Example \ref{ex:Roe.module} for the mildly less general case where $E$ is a trivial line bundle, but with $X$ allowed to be a general proper metric space with a proper cocompact isometric $\Gamma$-action \cite{Roe}. One motivation of \cite{Roe} was to connect \emph{coarse geometry} to the \emph{analytic assembly map} in the arena of the Baum--Connes conjectures; in particular, Lemma 2.3 of \cite{Roe} provides an identification of the compact operators on this Hilbert $C^*$-module with the \emph{equivariant Roe algebra} $C^*(X,\Gamma)$ (completion in $\cL(L^2(X))$ of the $\Gamma$-equivariant locally compact, finite propagation operators). In a subsequent work of the authors \cite{Ludewig-Thiang}, the algebra $C^*(X,\Gamma)$ plays a prominent role in proving, via coarse index methods, the robustness of edge states for Chern topological insulators of generic shapes $U\subset X=\RR^2$. 

One important role of the bundle $E$ is to allow for gauge (associated) bundles, especially in the context of electromagnetism (with ${\rm U}(1)$ as local gauge group) and quantum Hall Hamiltonians/magnetic Schr\"{o}dinger operators. Here, the $\Gamma$-action on $X$ is allowed to be lifted to a \emph{projective} unitary representation on $L^2(X,E)$, as encoded by a ${\rm U}(1)$-valued group 2-cocycle $\sigma$. For free, properly discontinuous, cocompact $\Gamma$-actions on a Riemannian manifold $X$, Gruber (\S V.B of \cite{Gruber}) made constructions almost identical to Example \ref{ex:Roe.module}, modified slightly to the projective-linear context, to produce inside $L^2(X,E)$ a standard Hilbert $C^*$-module over the \emph{twisted} group algebra $C^*_r(\Gamma,\sigma)$. Note that $C^*_r(\Gamma,\sigma)$ is generally noncommutative, even if $\Gamma$ is an abelian group such as $\ZZ^d$ (see \S \ref{sec:twisted.algebras}). The paper \cite{Gruber} uses Hilbert $C^*$-module language to study noncommutative Bloch theory (also studied in \cite{BS,MM}), with the intention to understanding some spectral properties of abstract $\mathcal{A}$-elliptic operators on Hilbert $\mathcal{A}$-modules, with $\mathcal{A}$ a general unital $C^*$-algebra. 

The commutative Bloch theory setting also applies to $C^*_r(\ZZ^d,\sigma)$ under a so-called \emph{rational flux} condition --- for such rational cocycles $\sigma$, the algebra $C^*_r(\ZZ^d,\sigma)$ is Morita equivalent to $C^*_r(\ZZ^d)$. For example, the discrete magnetic translations on $X=\RR^d$ corresponding to rational fluxes per unit cell will generate such a $C^*_r(\ZZ^d,\sigma)$. A suitable superlattice (a subgroup of $\ZZ^d$ whose fundamental domain has integer flux though it) will act linearly rather than projective linearly, and commutative Bloch theory can be carried out with respect to this superlattice. In this context, the Bloch--Floquet transform converts $L^2(\RR^d)$ into a Hilbert bundle over $\TT^d=\widehat{\ZZ}^d$, and topological aspects of its spectral subbundles (sometimes called \emph{Bloch bundles}) as modules over $C(\TT^d)$ is studied in \cite{DNPanati}, especially \S7 therein. The availability of a \emph{bona fide} topological space $\TT^d$ whose bundle theory mirrors the module theory of $C^*_r(\Gamma)$ is very useful for physics applications, including finding the strongest statement of the so-called localisation-topology dichotomy/correspondence (e.g.\ \cite{Monaco}). 

Our focus, however, is on the setting in which no such topological space exists for $C^*_r(\Gamma)$ or $C^*_r(\Gamma,\sigma)$, yet we can state a reasonable form of the localisation-\emph{noncommutative} topology dichotomy/correspondence (e.g.\ Theorem \ref{MainTheoremIntro}, and examples of applications in \S \ref{sec:physics}).

\section{Finitely generated projective pre-Hilbert modules}\label{sec:fgp.pre.Hilbert}

For a unital $C^*$-algebra $A$, it is known, cf.\ Thm.~15.3.8 of \cite{WO} that there is a correspondence between finitely generated projective (f.g.p.) $A$-modules and projections in $\cK(H)\otimes A$, the algebra of compact adjointable module maps on the Hilbert $A$-module $H \otimes A\cong \ell^2(A)$, c.f.~Example~\ref{L2Module}. Explicitly, this means that given a projection $p\in\cK(H)\otimes A$, one has that the Hilbert $A$-submodule $p(H \otimes A)$ is isomorphic to $q(A^n)$ for some $q\in \mathrm{M}_n(A), n\in\NN$.

For the dense subalgebra $H^\infty(\Gamma) \subseteq C^*_r(\Gamma)$, we can achieve a similar result for pre-Hilbert $H^\infty(\Gamma)$-modules. 

\begin{lemma}\label{LemmaSpectralProjective}
Let $\Gamma$ be a finitely generated group of polynomial growth.
Let $p$ be a projection in  $H^\infty(\Gamma, \cK(H)) \subset \cK(H)\otimes C^*_r(\Gamma)$, where $H$ is a separable Hilbert space. Then $p(H^\infty(\Gamma, H))$ is isomorphic as a pre-Hilbert $H^\infty(\Gamma)$-module to a f.g.p.\ $H^\infty(\Gamma)$-module.
\end{lemma}

\begin{proof}
By Prop.~\ref{PropSpectral}, $H^\infty(\Gamma, A)$ is spectral and closed under holomorphic functional calculus in $A \otimes C^*_r(\Gamma)$ for $A = \mathrm{M}_n(\CC)$ as well as $A = \cK^+(H)$, the $C^*$-algebra of compact operators on $H$ with a unit adjoined. The proof is then essentially an adaptation of Lemma~15.4.1 of \cite{WO}. 

We may assume that $H = \ell^2(\NN)$ and write $\cK = \cK(H)$. Let $\pi_n \in \cK \otimes C^*_r(\Gamma)$ denote the standard rank $n$ projection in $\cK$, tensored with the identity in $C^*_r(\Gamma)$. As $\{\pi_n\}_{n\in\NN}$ is an approximate identity for $\cK\otimes C^*_r(\Gamma)$, after writing $q_n\coloneqq \pi_np \pi_n$ for the truncation (no longer a projection in general but still self-adjoint), we can arrange for 
\begin{equation*}
\|p-q_n\|\ <\epsilon\leq\frac{1}{12}. 
\end{equation*}
Furthermore, $q_n$ is positive with $\|q_n\|\leq 1$, so that 
\begin{equation*}
\|q_n^2-q_n\|\leq \|q_n(q_n-p)\| + \|(q_n-p)p\| + \|p-q_n\| <3\epsilon\leq\frac{1}{4}.
\end{equation*}
Note that $q_n$ can be regarded as an element of ${\rm M}_n(H^\infty(\Gamma))$. If $\lambda\in\RR$ is in the spectrum of $q_n$ (which can be taken in ${\rm M}_n(H^\infty(\Gamma))$ or in ${\rm M}_n(C^*_r(\Gamma))$ since  the first is spectral in the latter), then by the above calculation, $0\leq \lambda-\lambda^2<\epsilon\leq\frac{1}{4}$. This implies that
\begin{equation*}
\lambda\in \bigl[0,\tfrac{1}{2}-\delta\bigr)\cup\bigl(\tfrac{1}{2}+\delta,1\bigr], \quad \text{where} \quad \delta=\frac{1}{2}\sqrt{1-4\epsilon},
\end{equation*}
so $q$ has a spectral gap at $\tfrac{1}{2}$.
The function $f = f(\mu)$ which is 0 to the left of $\mu = \frac{1}{2}$ and 1 to the right is therefore holomorphic in a neighbourhood of the spectrum of $q_n$, so we obtain a projection $q=f(q_n)\in {\rm M}_n(H^\infty(\Gamma))\subset H^\infty(\Gamma, \cK) \subset\cK\otimes C^*_r(\Gamma)$ which satisfies $q \leq \pi_n$ by construction (that $q \in \mathrm{M}_n(H^\infty(\Gamma))$ follows because $\mathrm{M}_n(H^\infty(\Gamma))$ is closed under holomorphic functional calculus). Moreover, we have
\begin{equation*}
 \|q_n - q\| = \max_{\mu \in \mathrm{Spec}(q_n)} |\lambda - f(\lambda)| < \frac{1}{2}-\frac{1}{2}\sqrt{1-4\epsilon} \leq \frac{1}{8},
\end{equation*}
by the choice of $\epsilon < \tfrac{1}{12}$. Therefore
\begin{equation*}
\|p-q\|\leq \|p-q_n\| + \|q_n-q\|< \frac{1}{12}+\frac{1}{8} <1,
\end{equation*}
and a standard construction (e.g.\ Prop.\ 5.2.6 in \cite{WO}) gives 
\begin{equation*}
z=(2q-1)(2p-1)+1\in H^\infty(\Gamma, \cK^+) \subset \cK^+\otimes C^*_r(\Gamma)
\end{equation*}
implementing $q z=zp$, where $\cK^+$ is the unital $C^*$-algebra given by adjoining a unit to $\cK$.
 Then $\|z-2\|=\|2(2q-1)(p-q)\|<2$, so that $z$ is invertible in $\cK^+\otimes C^*_r(\Gamma)$. In fact, we have $z^{-1}\in H^\infty(\Gamma, \cK^+)$, as by Prop.~\ref{PropSpectral}, $H^\infty(\Gamma, \cK^+)$ is spectral in $\cK \otimes C^*_r(\Gamma)$. Note that $zp=q z$ gives $z^*q=p z^*$ and also $p(z^*z)=(z^*z)p$, so with $u=z(z^*z)^{-1/2}$ being the unitary in the polar decomposition of $z$, which again exists by Prop.~\ref{PropSpectral}, we obtain a unitary equivalence $up u^{-1}=q$ in $H^\infty(\Gamma, \cK^+)$. 
 
 The unitary $u \in H^\infty(\Gamma, \cK^+)$ constructed above gives an isomorphism
 \begin{equation*}
   \Phi: p(H^\infty(\Gamma, H)) \longrightarrow q(H^\infty(\Gamma)^n) \subset q\bigl(H^\infty(\Gamma, H)\bigr), \quad x \longmapsto u x.
 \end{equation*}
 Note that since $x \in p(H^\infty(\Gamma, H))$, we have $x = px = u^{-1}qux$, hence $\Phi(x) = qux = q(qux) \in q(H^\infty(\Gamma)^n)$. Clearly, $\Phi$ commutes with right multiplication by $H^\infty(\Gamma)$, hence is a module map. Moreover, since $u$ is unitary, we have
 \begin{equation*}
   \bigl(\Phi(x)\bigl|\Phi(y)\bigr) = (ux|uy) = (x|u^*u y) = (x|y),
 \end{equation*}
 hence $\Phi$ preserves the inner product. This finishes the proof.
 \end{proof}

Of course, Lemma \ref{LemmaSpectralProjective} also implies that $p(H \otimes C^*_r(\Gamma))$ is unitarily isomorphic as a Hilbert $C^*_r(\Gamma)$-module to a direct summand of $C^*_r(\Gamma)^n$. A priori however, the minimal required $n$ could be smaller in the $C^*_r(\Gamma)$-module case than in the $H^\infty(\Gamma)$-module case. The next lemma shows that this is not the case.

\begin{lemma}\label{projmodlemma}
Let $p$ be a projection in $H^\infty(\Gamma, \cK(H)) \subset \cK(H)\otimes C^*_r(\Gamma)$ as in Lemma \ref{LemmaSpectralProjective}.
If $p(H \otimes C^*_r(\Gamma))$ is isomorphic as a Hilbert $C^*_r(\Gamma)$-module to a direct summand of $C^*_r(\Gamma)^n$ for some $n\in\NN$, then $p(H^\infty(\Gamma, H))$ is isomorphic as a pre-Hilbert $H^\infty(\Gamma)$-module to a direct summand of $H^\infty(\Gamma)^n$ for the same $n$.
\end{lemma}

In particular, $p(H^\infty(\Gamma, H))$ is (isomorphic to) a pre-Hilbert $H^\infty(\Gamma)$-module of rank $n$ if and only if $p(H \otimes C^*_r(\Gamma))$ is (isomorphic to) a free Hilbert $C^*_r(\Gamma)$-module of rank $n$. 

\begin{proof}
Lemma \ref{LemmaSpectralProjective} gives $p(H^\infty(\Gamma, H))\cong q(H^\infty(\Gamma)^n)$ and $p(H \otimes C^*_r(\Gamma))\cong q(C^*_r(\Gamma)^n)$ for some $n\in\NN$, with $q$ a projection in ${\rm M}_n(H^\infty(\Gamma))\subset {\rm M}_n(C^*_r(\Gamma))$. Now suppose a unitary $y\in {\rm M}_n(C^*_r(\Gamma))$ implements a further reduction 
\begin{equation*}
yq y^{-1}=\begin{pmatrix} \check{q} & 0 \\ 0 & 0 \end{pmatrix} \quad \text{with} \quad \check{q}\in {\rm M}_{n-1}(C^*_r(\Gamma)).
\end{equation*} 
Choose a unitary $y_0 \in {\rm M}_n(H^\infty(\Gamma))$ close to $y$ so that $y_0 q y_0^{-1}$ is a projection in $\mathrm{M}_n(H^\infty(\Gamma))$ close to $yq y^{-1}={\rm diag}(\check{q},0)$. 

On the other hand, pick $\check{r}_0=\check{r}_0^*\in {\rm M}_{n-1}(H^\infty(\Gamma))$ (not necessarily a projection) close to $\check{q}\in {\rm M}_{n-1}(C^*_r(\Gamma))$ such that its spectrum avoids the line ${\rm Re}(\mu)=\frac{1}{2}$. 
As in the proof of Lemma~\ref{LemmaSpectralProjective}, holomorphic functional calculus gives a projection $\check{q}_0 :=f(\check{r}_0)\in {\rm M}_{n-1}(H^\infty(\Gamma))$ close to $\check{q}$. 

Now the projections $q'=y_0 q y_0^{-1}$ and $q''={\rm diag}(\check{q}_0,0)$ are in ${\rm M}_n(H^\infty(\Gamma))$ and are both close to ${\rm diag}(\check{q},0)$, hence close to each other; therefore, they are unitarily equivalent in ${\rm M}_n(H^\infty(\Gamma))$ (take the polar decomposition of the invertible $(2q''-1)(2q'-1)+1\in {\rm M}_n(H^\infty(\Gamma))$). This gives a unitary $u \in \mathrm{M}_n(H^\infty(\Gamma))$ such that
\begin{equation*}
  u (y_0 q y_0^{-1})u^{-1} = \begin{pmatrix} \check{q}_0 & 0 \\ 0 & 0 \end{pmatrix},
\end{equation*}
which finishes the proof.
\end{proof}

The above results combine to the following corollary.

\begin{corollary} \label{CorollaryOfProjModLemma}
Let $p$ be a projection in $H^\infty(\Gamma, \cK(H)) \subset \cK(H)\otimes C^*_r(\Gamma)$ as in Lemma \ref{LemmaSpectralProjective}, and let $P = p(H \otimes C^*_r(\Gamma))$ be the corresponding $C^*_r(\Gamma)$-module. Then $P$ is a finitely generated and projective Hilbert $C^*_r(\Gamma)$-module. Moreover, for each $n \in \NN$, there exists a Hilbert $C^*_r(\Gamma)$-module $Q$ such that
  \begin{equation*}
    P \oplus Q \cong C^*_r(\Gamma)^n
  \end{equation*}
  if and only if there exists a pre-Hilbert $H^\infty(\Gamma)$-module $Q^\prime$ such that
  \begin{equation*}
    \bigl(P \cap H^\infty(\Gamma, H)\bigr) \oplus Q^\prime \cong H^\infty(\Gamma)^n.
  \end{equation*}
\end{corollary}

\section{Rapidly decaying functions and admissible operators}\label{sec:submodules.section}

Let $X$ be a complete Riemannian manifold and let $E$ be a Hermitian vector bundle over $X$. In this section, we consider the following function spaces inside $L^2(X, E)$, whose definitions refer to a choice of basepoint $x_0 \in X$, but are in fact independent of it. Later, we will relate these function spaces under the assumption of \emph{polynomial growth}.

\begin{definition}\label{defn:Schwartzclass}
  For some fixed $x_0 \in X$, we define
  \begin{equation*}
  \begin{aligned}
     L^2_{\mathrm{rd}}(X, E) &:= \bigl\{ w \in L^2(X, E) ~\bigl|~ \forall n \in \NN\; \exists C>0: \|w\|_{L^2(X \setminus B_{R, x_0}, E)} \leq C R^{-n}\bigr\}\\
          \cS(X, E) &:= \bigl\{ w \in C^\infty(X, E) ~\bigl|~ \forall n, m \in \NN\; \exists C>0: |\nabla^mw(x)| \leq C d(x_0, x)^{-n}\bigr\}
   \end{aligned}
  \end{equation*}
The first is the space of $L^2$-sections of $E$ having {\em rapid decay}, while the second one is the space of $E$-valued {\em Schwartz functions} on $X$.
\end{definition}

If we are in our Basic Setup from the introduction, which we will assume from now on, the Milnor--\v{S}varc lemma applies \cite[Thm.~8.37]{GGT}, so $\Gamma$ must be finitely generated, and, as a metric space with its word metric, be quasi-isometric to $X$. In other words, there exist constants $0 < \varepsilon < C$ such that
\begin{equation} \label{MetricEquivalence}
  \varepsilon\bigl(1+ L(\gamma)\bigr) \leq 1+ d(x_0, \gamma x_0) \leq C \bigl(1+ L(\gamma)\bigr)
\end{equation}
for all $\gamma \in \Gamma$ and all $x_0 \in X$. Of course, these constants (as well as the length function) depend on the choice of finite generating set $\mathscr{S}$ for $\Gamma$ (c.f.~Section~\ref{SectionRapidDecay}). We moreover assume that $X$ has {\em polynomial volume growth}, meaning that there exists $s \geq 0$ such that for some (equivalently, for any) $x_0 \in X$, we have
\begin{equation} \label{PolynomialGrowthX}
    \mathrm{vol}\bigl(B_{R, x_0}\bigr) < C (1+R)^s
\end{equation}
for some constant $C>0$ and all $R>0$, where $B_{R, x_0}$ is the ball around $x_0$ of radius $R$. The estimates \eqref{MetricEquivalence} imply that this condition is equivalent to the condition that $\Gamma$ has polynomial growth in the sense of \eqref{PolyGrowth}. The estimate \eqref{PolynomialGrowthX} now implies the inclusions
\begin{equation} \label{Inclusions}
  \cS(X, E) \subseteq L^2_{\mathrm{rd}}(X, E) \subseteq L^2_\Gamma(X, E);
\end{equation}
in fact, the second inclusion follows from the following lemma.

\begin{lemma}
For any bounded fundamental domain, the map $\Phi_{\cF}$ defined in \eqref{IsoPhi} restricts to an isomorphism of pre-Hilbert $H^\infty(\Gamma)$-modules
\begin{equation*}
  \Phi_\cF: L^2_{\mathrm{rd}}(X, E) \stackrel{\cong}{\longrightarrow} H^\infty(\Gamma, H) \subset H \otimes C^*_r(\Gamma).
\end{equation*}
where $H = L^2(\cF, E)$.
\end{lemma}

\begin{proof}
For $R \geq 0$, let $\Gamma_R$ be the set of all $\gamma \in \Gamma$ such that $\gamma\cF$ is not included in $B_{R, x_0}$. Then for $w \in L^2(X, E)$ and any $s \geq 0$, we have the estimate
\begin{equation*}
  \|w\|_{L^2(X \setminus B_{R, x_0}, E)}^2 \leq \sum_{\gamma \in \Gamma_R} \|\gamma^* w\|_{H}^2 \leq \bigl\|\Phi_\cF(w)\bigr\|_{s, H}^2 \cdot \sup_{\gamma \in \Gamma_R} \bigl(1+L(\gamma)\bigr)^{-2s}.
\end{equation*}
Together with \eqref{MetricEquivalence}, this shows that if $\Phi_\cF(w) \in H^\infty(\Gamma, H)$, then $w \in L^2_{\mathrm{rd}}(X, E)$.

To see the converse, note that for $\gamma \in \Gamma \setminus \Gamma_R$, we have $\gamma \cF \subset B_{R, x_0}$, hence \eqref{MetricEquivalence} implies that
\begin{equation*}
  1+L(\gamma) \leq C\bigl(1 + d(x_0, \gamma x_0)\bigr) \leq C^\prime (1 + R);
\end{equation*}
here we assume that $x_0 \in \cF$. Using this, we obtain for any $w \in L^2(X, E)$ and any $s \geq 0$
\begin{equation*}
\begin{aligned}
  \bigl\|\Phi_\cF(w)\bigr\|_{s, H}^2 &\leq \sum_{m=0}^\infty \sum_{\gamma \in \Gamma_{m} \setminus \Gamma_{m+1}} \|w\|_{L^2(\gamma \cF, E)}^2 \bigl(1+L(\gamma)\bigr)^{2s} \\
  &\leq C \sum_{m=1}^\infty \bigl(2+m\bigr)^{2s} \|w\|_{L^2(X \setminus B_{m, x_0}, E)}^2.
  \end{aligned}
\end{equation*}
If now $w \in L^2_{\mathrm{rd}}(X, E)$, this sum converges, showing that $\Phi_\cF(w) \in H^\infty(\Gamma, H)$.
\end{proof}

Wrapping up the results from above, with respect to the choice of a bounded fundamental domain, we have the following diagram, where each of the horizontal arrows, given by the map $\Phi_\cF$, is an isomorphism preserving the respective (pre-)Hilbert module structures.
\begin{equation} \label{IsoPhiDiagram}
  \begin{tikzcd}[column sep = 1.7cm]
    L^2(X, E) \ar[r, "\Phi_{\cF}"] & L^2(\cF, E) \otimes \ell^2(\Gamma) \\
    L^2_\Gamma(X, E) \ar[u, hookrightarrow] \ar[r, "\Phi_{\cF}"] & L^2(\cF, E) \otimes C^*_r(\Gamma) \ar[u, hookrightarrow]  \\
    L^2_{\mathrm{rd}}(X, E) \ar[u, hookrightarrow] \ar[r, "\Phi_{\cF}"] & H^\infty\bigl(\Gamma, L^2(\cF, E)\bigr) \ar[u, hookrightarrow]  \\
    L^2_c(X, E) \ar[u, hookrightarrow] \ar[r, "\Phi_{\cF}"] & L^2(\cF, E) \otimes \CC[\Gamma] \ar[u, hookrightarrow]  \\
   \end{tikzcd}
\end{equation}

\begin{remark}
Under the condition of polynomial growth, $H^\infty(\Gamma)$ is in fact a nuclear space \cite[Thm.~3.1.7]{Jolissaint}, hence the algebraic tensor product with any other Banach space has a unique tensor product topology. In particular, the tensor product $L^2(\cF, E) \otimes H^\infty(\Gamma)$ is unambiguously defined and equal to $H^\infty\bigl(\Gamma, L^2(\cF, E)\bigr)$. However, we will not need this fact.
\end{remark}

\begin{definition} \label{DefAdmissible}
  We say that an operator $A \in \cL(L^2(X, E))$ is {\em admissible} if it has a smooth integral kernel $a(x, y)$ which is rapidly decaying in the sense that for $n, m \in \NN$ and each $s \geq 0$, there exists a constant $C>0$ such that
\begin{equation} \label{KernelEstimates}
  \bigl|\nabla_x^n \nabla_y^m a(x, y)\bigr|  < C \bigl(1 + d(x, y)\bigr)^{-s}
\end{equation}
for all $x, y \in X$. 
\end{definition}

\begin{lemma} \label{LemmaMapsToSchwartz}
If $A \in \cL(L^2(X, E))$ is an admissible operator and $w \in L^2_{\mathrm{rd}}(X, E)$, then $A(w) \in \cS(X, E)$.
\end{lemma}

\begin{proof}
Let $a(x, y)$ be the integral kernel of $A$ and let $w \in L^2_{\mathrm{rd}}(X, E)$. We have
\begin{equation*}
  \nabla^m Aw(x) = \int_X \nabla_x^m a(x, y) w(y) \dd y,
\end{equation*}
which is absolutely convergent by the polynomial growth condition on $X$ and the decay of $a(x, y)$. Now for any $s \geq 0$, we have
\begin{equation*}
  \bigl|\nabla^n Aw(x)\bigr| \leq \int_X \bigl| \nabla_x^n a(x, y)\bigr| |w(y)| \dd y \leq C\int_X \bigl(1+d(x, y)\bigr)^{-s}\bigl(1+d(y, x_0)\bigr)^{-s} \dd y.
\end{equation*}
By the triangle inequality, we have for any $r \geq 0$ that
\begin{equation*}
  \frac{d(x, x_0)^r}{\bigl(1+d(x, y)\bigr)^s\bigl(1+d(y, x_0)\bigr)^s} \leq \bigl(1+d(y, x_0)\bigr)^{r-s},
\end{equation*}
which is integrable with respect to $y$ for $s \geq 0$ large enough, since $X$ has polynomial growth. Thus for any $n \in \NN$ and $r\geq 0$, $|\nabla^n Aw(x)|d(x, x_0)^r$ is bounded independent from $x$, which was what we needed to show.
\end{proof}

\begin{proposition} \label{PropRapidProjective}
Let $p \in \cL(L^2(X, E))$ be an admissible $\Gamma$-invariant projection. Then $P = p(L^2_\Gamma(X, E))$ is a f.g.p.\ Hilbert $C^*_r(\Gamma)$ submodule of $L^2_\Gamma(X, E)$. Moreover, given $n \in \NN$, there exists a f.g.p.\ Hilbert $C^*_r(\Gamma)$-module $Q$ such that
 \begin{equation*}
   P \oplus Q \cong C^*_r(\Gamma)^n
 \end{equation*}
if and only if there exists there exists a f.g.p.\ pre-Hilbert $H^\infty(\Gamma)$-module $Q^\prime$ such that 
\begin{equation*}
  P \cap \cS(X, E) \oplus Q^\prime \cong H^\infty(\Gamma)^n.
\end{equation*}
\end{proposition}

\begin{proof}
With a view on \eqref{IsoPhiDiagram}, the result follows from Corollary~\ref{CorollaryOfProjModLemma}, if we can show that 
\begin{equation} \label{EqLemmaInclusion}
  \Phi_\cF \,p \, \Phi_\cF^{-1} \in H^\infty\bigl(\Gamma, \cK(H)\bigr), \quad \text{for} \quad H = L^2(\cF, E).
\end{equation}
for some bounded fundamental domain $\cF$. Corollary~\ref{CorollaryOfProjModLemma} states that the existence of a f.g.p.\ Hilbert $C^*_r(\Gamma)$-module $Q$ as above implies the existence of a f.g.p.\ pre-Hilbert $H^\infty(\Gamma)$-module $Q^\prime$ such that $P \cap L^2_{\mathrm{rd}}(X, E) \oplus Q^\prime \cong H^\infty(\Gamma)^n$ and vice versa. 

However, if for $w \in P \cap L^2_{\mathrm{rd}}(X, E)$, then $w = p(v) \in L^2_{\mathrm{rd}}(X, E)$ with $v \in L^2_\Gamma(X, E)$. Hence, since $p$ is a projection, $w = p(p(v))$, which by Lemma~\ref{LemmaMapsToSchwartz} is contained in $\cS(X, E)$, as $p(v) \in L^2_{\mathrm{rd}}(X, E)$. In other words, we have
\begin{equation*}
  P \cap L^2_{\mathrm{rd}}(X, E) = P \cap \cS(X, E).
\end{equation*}
To show \eqref{EqLemmaInclusion}, let $p(x, y)$ be the integral kernel of $p$. Fix $f \in L^2(\cF, E)$, and let $w = \Phi_F^{-1}(f \otimes \gamma^{-1})$. We have
\begin{equation*}
\begin{aligned}
  p(w)(x) &= \int_{\gamma \cF} p(x, y) f(\gamma^{-1}y) \dd y = \int_\cF p(\gamma^{-1}x,  y)f(y) \dd y,
\end{aligned}
\end{equation*}
where we used that the integral kernel of $p$ is $\Gamma$-invariant, $p(\gamma x, \gamma y) = p(x, y)$ for all $\gamma \in \Gamma$, $x, y \in X$. Hence
\begin{equation*}
\begin{aligned}
  \Phi_\cF\,p(w)(x) &= \sum_{\rho \in \Gamma} \int_\cF p( \gamma^{-1} \rho x, y) f(y) \dd y \otimes \rho^{-1}\\
  &= \sum_{\rho \in \Gamma} \int_\cF p(\rho x, y) f(y) \dd y \otimes \rho^{-1} \gamma^{-1}.
\end{aligned}
\end{equation*}
In other words, we have 
\begin{equation*}
  \Phi_\cF \,p \,\Phi_\cF^{-1} = \sum_{\gamma \in \Gamma} A_\gamma \otimes \gamma^{-1}, ~~~~\text{with}~~~~(A_\gamma f)(x) = \int_\cF p(\gamma x, y) f(y) \dd y,
\end{equation*}
which is a compact operator, as it has a bounded integral kernel. For the operator norm of $A_\gamma$, we have the estimate
\begin{equation*}
  \|A_\gamma\|_{\mathrm{op}} \leq \mathrm{vol}(\cF) \sup_{x, y \in \cF} \bigl| p(\gamma x, y)\bigr| \leq C \bigl(1+d(\gamma x, y)\bigr)^{-s} = C^\prime \bigl(1+L(\gamma)\bigr)^{-s},
\end{equation*}
for any $s \geq 0$, where we used that $p$ is admissible, and that $\cF$ is bounded, together with \eqref{MetricEquivalence}. This follows from the fact that $p$ is admissible and that $X$ is quasi-isometric to $\Gamma$. This implies \eqref{EqLemmaInclusion} since $\Gamma$ has polynomial growth.
\end{proof}

\section{The main theorem}\label{main.theorem.section}

Let $X,E,\Gamma$ be as in the Basic Setup of the introduction. Thm.~\ref{MainTheoremIntro} from the introduction is a special case of the following result, Thm.~\ref{TheoremMainTheorem}, combined with Thm.~\ref{ThmBlackBox} further below.

Recall that a {\em tight frame} of a Hilbert space $H$ is a collection $f_1, f_2, \dots$ of (possibly linearly dependent) elements of $H$ such that for all $h\in H$, we have
\begin{equation*}
   \|h\|_H^2 = \sum_{j=1}^\infty |\langle h, f_j\rangle|^2.
\end{equation*}

\begin{theorem} \label{TheoremMainTheorem}
Suppose that $\Gamma$ has polynomial growth and let $D$ be a $\Gamma$-invariant, self-adjoint (unbounded) operator on $L^2(X, E)$ such that $\psi(D)$ is admissible for each Schwartz function $\psi \in \cS(\RR)$. Let $S$ be a compact subset of the spectrum of $D$ which is separated from the rest of the spectrum and let $L^2_S(X, E)$ be the corresponding spectral subspace. Set 
\begin{equation*}
P_S := L^2_S(X, E) \cap L^2_\Gamma(X, E).
\end{equation*}
\begin{enumerate}
\item[(i)] $P_S$ is a f.g.p.\ Hilbert $C^*_r(\Gamma)$-module.
\item[(ii)] Suppose that there exists a f.g.p.\ Hilbert $C^*_r(\Gamma)$-module $Q$ such that $P_S \oplus Q \cong C^*_r(\Gamma)^n$. Then there exist $w_1, \dots, w_n \in \cS(X, E)\cap P_S$ 
such that these functions and their translates form a tight frame of $L^2_S(X, E)$. Moreover, if $Q = \{0\}$, then we can arrange for this tight frame to be an orthonormal basis.
\item[(iii)] Conversely, suppose that there exist $w_1, \dots, w_n \in P_S$ that together with their $\Gamma$-translates form a tight frame of $L^2_S(X, E)$. Then there exists a f.g.p.\ Hilbert $C^*_r(\Gamma)$-module $Q$ such that $P_S \oplus Q \cong C^*_r(\Gamma)^n$, and if the $w_1, \dots, w_n$ together with their translates are linearly independent, then $P_S$ is a free Hilbert $C^*_r(\Gamma)$-module.
\end{enumerate}
\end{theorem}

\begin{remark}
  Statement (iii) above is false if $w_1, \dots, w_n\in L^2_S(X, E)$ are not constrained to be in $P_S$. Indeed, there are examples where there exists an orthonormal basis of $L^2_S(X, E)$ consisting of $w_1, \dots, w_n$ and their translates, without $P_S$ being free, see the $\Gamma=\ZZ^2$ case in \S\ref{sec:nonabelian.examples}. However, Thm.~\ref{TheoremMainTheorem} states that for a basis comprising $w_1, \dots, w_n$ and their translates, if the $w_j$ do satisfy the mild decay property of being contained in $L^2_\Gamma(X, E)$ (c.f.~Remark~\ref{RemarkDecay}), then $P_S$ is free, and one can further choose the $w_j$ from the Schwartz space $\mathcal{S}(X, E)$.
\end{remark}

\begin{proof}
Let $p_S$ be the spectral projection in $L^2(X, E)$ associated to the subset $S \subset \mathrm{spec}(D)$. Since $S$ is separated from the rest of the spectrum, we can write $p_S = \psi(D)$ using functional calculus, where $\psi$ is a compactly supported smooth function such that $\psi(\lambda) = 1$ for $\lambda \in S$ and $\psi(\lambda) = 0$ whenever $\lambda \in \mathrm{spec}(D) \setminus S$. Thus $p_S$ is admissible in the sense of Def.~\ref{DefAdmissible}, and by Prop.~\ref{PropRapidProjective}, 
\begin{equation*}
P_S = L^2_S(X, E) \cap L^2_\Gamma(X, E) = p_S\bigl(L^2_\Gamma(X, E)\bigr)
\end{equation*}
 is finitely generated and projective, which is claim (i).

To show claim (ii), let $Q$ be a f.g.p.\ Hilbert $C^*_r(\Gamma)$-module such that $P_S \oplus Q \cong C^*_r(\Gamma)^n$. By Prop.~\ref{PropRapidProjective}, there exists a pre-Hilbert $H^\infty(\Gamma)$-module $Q^\prime$ and an isomorphism of pre-Hilbert $H^\infty(\Gamma)$-modules
\begin{equation*}
  \alpha:  H^\infty(\Gamma)^n = \CC^n \otimes H^\infty(\Gamma) \longrightarrow P_S \cap H^\infty(\Gamma, E) \oplus Q^\prime.
\end{equation*}
By Lemma~\ref{LemmaMapsToSchwartz}, we have $P_S \cap H^\infty(\Gamma, E) = P_S \cap \cS(X, E)$.
Define $w_1, \dots, w_n \in P_S \cap \cS(X, E)$ by $w_j := \mathrm{pr}_1(\alpha(e_j \otimes \mathbf{1}))$, where $\mathrm{pr}_1$ is the projection onto the first factor $P_S \cap \cS(X, E)$ and $e_j \in \CC^n$ is the $j$-th unit vector. Since $\alpha$ (being an isomorphism of pre-Hilbert $H^\infty(\Gamma)$-modules) preserves the inner products, the vectors $\gamma^*w_j = \alpha(e_j \otimes \gamma)$ satisfy 
\begin{equation} \label{ONB}
\bigl(\alpha(e_i \otimes \rho) \bigl| \alpha(e_j \otimes \gamma) \bigr) = (e_i \otimes \rho \,|\, e_j \otimes \gamma) = \delta_{ij} \rho^* \gamma.
\end{equation}
Composing the $H^\infty(\Gamma)$-valued inner product with the standard trace on $H^\infty(\Gamma)$ gives a scalar product on $P_S \cap \cS(X, E) \oplus Q^\prime$, and by \eqref{ONB}, the vectors $\gamma^*w_j$, $j = 1, \dots, n$, $\gamma \in \Gamma$, form an orthonormal basis of the completion with respect to this inner product. 

The inner product on $P_S\cap \cS(X, E) \subset L^2(X, E)$ obtained this way coincides with the restriction of the standard inner product on $L^2(X, E)$, and since $\cS(X, E)$ is dense in $L^2(X, E)$, the completion of $P_S\cap \cS(X, E)$ is $L^2_S(X, E)$. Because the orthogonal projection of an orthonormal basis to a subspace forms a tight frame of the subspace, this shows that the sections $\gamma^*w_j = \mathrm{pr}_1(\alpha(e_j \otimes \gamma))$, $j=1, \dots, n$, $\gamma \in \Gamma$ form a tight frame of $L^2_S(X, E)$. Clearly, if $Q = \{0\}$, then the $\gamma^*w_j$ are linearly independent, as then $\gamma^* w_j = \alpha(e_j \otimes \gamma)$, the $e_j \otimes \gamma$ are linearly independent and $\alpha$ is a vector space isomorphism.

To show (iii), let $w_1, \dots, w_n \in P_S$ be such that $\gamma^*w_j$, $j=1, \dots, n$, $\gamma \in \Gamma$, forms a tight frame of $L^2_S(X, E)$. From the characterization of tight frames \cite{Han}, there exists a Hilbert space $H$ and an orthonormal basis $v_{j, \gamma}$, $j=1, \dots, n$, $\gamma \in \Gamma$ of $L^2_S(X, E) \oplus H$ such that $\gamma^*w_j = \mathrm{pr}_1(v_{j, \gamma})$. Setting
\begin{equation*}
  v_{j, \gamma} \cdot \rho := v_{j, \gamma \rho}
\end{equation*}
for $\rho \in \Gamma$ defines a right representation of $C^*_r(\Gamma)$ on $L^2_S(X, E) \oplus H$. This action preserves $H$; in fact, if we let $v_{j, \gamma} = (\gamma^*w_j, v^\prime_{j, \gamma})$, then $v^\prime_{j, \gamma} \cdot \rho = v_{j, \gamma\rho}^\prime$. Setting
\begin{equation*}
  (v_{i, \gamma}^\prime | v_{j, \eta}^\prime) := \langle v_{i, \mathbf{1}}^\prime, v_{j, \mathbf{1}}^\prime\rangle_H \gamma^{-1}\eta
\end{equation*}
defines a $C^*_r(\Gamma)$-valued inner product on the subspace $Q_0 \subset H$ consisting of finite linear combinations of the $v^\prime_j, \gamma$, $j=1, \dots, n$, $\gamma \in \Gamma$. Letting $Q \subset H$ be the completion of $Q_0$ with respect to the norm induced from the inner product and the norm on $C^*_r(\Gamma)$ defines a Hilbert-$C^*_r(\Gamma)$-module such that $P_S \oplus Q \cong C^*_r(\Gamma)^n$, via the obvious isomorphism of Hilbert-$C^*(\Gamma)$-modules sending $v_{j, \gamma}$ to $e_j \otimes \gamma$. If the $\gamma^*w_j$ were in fact linearly independent, then $H = \{0\}$, so that $P_S \cong C^*_r(\Gamma)^n$.
\end{proof}

The following result shows that there are many examples of operators $D$ for which Thm.~\ref{TheoremMainTheorem} applies.

\begin{theorem} \label{ThmBlackBox}
Suppose that $\Gamma$ has polynomial growth.
Let $D$ be a self-adjoint, $\Gamma$-invariant differential operator acting on sections of $E$ and assume that either $D$ is elliptic of order one; or that $D$ is of order two and of Laplace type. Then for each Schwartz function $\psi \in \mathcal{S}(\RR)$, the operator $\psi(D)$ is admissible in the sense of Def.~\ref{DefAdmissible}.
\end{theorem}

Here by a {\em Laplace type} operator, we mean a second order operator such that in local coordinates $x= (x_1, \dots, x_n)$ on $X$, it is given by
\begin{equation*}
 D = -\sum_{i,j=1}^n g^{ij}(x) \frac{\partial^2}{\partial x_i \partial x_j} + \sum_{i=1}^n b_i(x) \frac{\partial}{\partial x_i} + c(x),
\end{equation*}
where $(g^{ij}(x))_{ij=1, \dots, n}$ is the inverse matrix of the coefficient matrix of the Riemannian metric on $X$ and $b_i$, $c$ are certain sections of the endomorphism bundle of $E$. Note that the self-adjointness requirement on $D$ poses some additional restrictions on the lower order coefficients $b_i$ and $c$.

\begin{proof}
The argument is similar to the one in \cite{CGT}. 
We will show that for all $\varphi \in \cS(\RR)$, $m, n, k \in \ZZ$, there exists a constant $C>0$ such that
\begin{equation} \label{WantToShow}
   \bigl\| \varphi(D)w\bigr\|_{H^k(X \setminus B_{R, K}, E)} < C R^{-n} \|w\|_{H^m(X, E)}
\end{equation} 
whenever $w$ is a smooth section of $E$ with support in a compact set $K \subset X$. Here $B_{R, K}$ denotes the set of $x \in X$ with $d(x, K) \leq R$ and for an open subset $Y\subset X$,  $H^k(Y, E) \subset L^2(Y, E)$ is the Sobolev space of sections in $E$ with square-integrable weak derivatives up to order $k$. The result then follows from the results in \cite{Engel}: The estimates \eqref{WantToShow} imply that $\varphi(D)$ is a quasi-local smoothing operator (c.f.\ Def.~2.14 ibid.); Sobolev embedding together with the results of Section~2.4 in \cite{Engel} shows that $\varphi(D)$ is admissible in the sense of Def.~\ref{DefAdmissible} above.

First consider the case that $D$ is of order one. In that case, the wave operator $e^{isD}$ has finite propagation, meaning that there exists a constant $\alpha>0$ such that whenever $w$ has support in a compact set $K \subset X$, then $e^{isD}w$ has support in $B_{\alpha s, K}$. We use the formula
\begin{equation} \label{FormulaFunctionOfOperator}
  \varphi(D)w = \frac{1}{2\pi} \int_{-\infty}^\infty \hat{\varphi}(s) e^{i sD} w \, \dd s
\end{equation}
for Schwartz functions $\varphi$, where $\hat{\varphi}$ denotes the Fourier transform of $\varphi$. Given $k, m \in \NN_0$, let $\varphi(x) := (1+x^2)^{(k-m)/2} \psi(x)$. Then by elliptic estimates, we have $\|\psi(D)w\|_{H^k} = \|\varphi(D) w\|_{H^m}$ (if one defines $H^k$ norms suitably). Let $\chi_R$ be the function that is equal to zero on $B_{R, K}$ and identically one on the complement. 
Now
\begin{equation*}
\begin{aligned}
   \bigl\| \psi(D)w\bigr\|_{H^k(X \setminus B_{R, K}, E)} &=  \bigl\| \chi_R \varphi(D)w\bigr\|_{H^m(X, E)} \\
   &\leq \frac{1}{2\pi} \int_{-\infty}^\infty|\hat{\varphi}(s)|  \bigl\|\chi_R e^{i sD} \varphi\bigr\|_{H^m(X, E)} \, \dd s\\
   &\leq \frac{1}{2\pi} \left(\int_{-\infty}^{-R} + \int_R^\infty\right) |\hat{\psi}(s)|  \bigl\| e^{i sD} \varphi\bigr\|_{H^m(X, E)} \, \dd s
\end{aligned}
\end{equation*}
by the finite propagation of $e^{isD}$.
This gives the estimate \eqref{WantToShow}, since the operator $e^{isD}$ is uniformly bounded independent of $s$ and $\hat{\psi}$ is rapidly decaying.

If $D$ is order two and of Laplace type, the spectrum of $D$ is bounded below. By possibly replacing $D$ by $D+\mu$, we may assume that the spectrum of $D$ is bounded below by $\varepsilon >0$. We may then take $\psi$ to be an even function, in which case formula \eqref{FormulaFunctionOfOperator}, applied for $\sqrt{D}$ instead of $D$ yields
\begin{equation*}
  \psi(D)w = \frac{1}{2\pi} \int_{-\infty}^\infty \hat{\varphi}(s) \cos(s \sqrt{D})w\, \dd s;
\end{equation*}
here $\hat{\varphi}$ is the Fourier transform of $\varphi(s) = \psi(s^2)$. The point is now that $\cos(s \sqrt{D})$ is the solution operator to the wave equation,  $(\tfrac{\partial^2}{\partial s^2} + D)w = 0$, which again has finite propagation speed \cite{Taylor1, BGP}. The proof is then similar to the argument before.
\end{proof}


\section{Application to good Wannier basis existence problem}\label{sec:physics}
We apply our results to the old problem of constructing well-localised \emph{Wannier bases} in solid state physics. Most existing results on Wannier bases apply to the basic case where $\Gamma=\ZZ^d$ is a lattice of translations acting on affine Euclidean space $X=\RR^d$. Fourier transform identifies $C^*_r(\ZZ^d)$ with $C(\TT^d)$, the continuous functions on 
\begin{equation*}
\TT^d={\rm Hom}(\ZZ^d,{\rm U}(1))\equiv\wh{\Gamma},
\end{equation*}
 the \emph{Brillouin zone/torus} in physics, and one studies $\Gamma$-invariant Hamiltonian operators $D=D^*$ acting on the Hilbert space $L^2(\RR^d)$ of quantum mechanical wavefunctions.
 
We will recast the idea of Wannier bases in noncommutative topology/geometry language, so that our results become applicable, e.g.\ to all crystallographic $\Gamma$ at once, and also in certain non-Euclidean settings as illustrated by our final example. It is instructive, however, to first recall the basic notions in the commutative case $\Gamma=\ZZ^d$, which proceeds via classical Bloch theory.

\subsection{Commutative Bloch--Floquet transform and Wannier bases}\label{sec:BlochFloquetreview}
The fundamental domain for the $\ZZ^d$ action on $\RR^d$ is an affine torus $T^d=\RR^d/\ZZ^d$ (not to be confused with the Brillouin torus $\TT^d)$.  The \emph{Bloch--Floquet} transform is a unitary map $\Psi:L^2(\RR^d)\ni f \mapsto \hat{f}\in L^2(\TT^d,L^2(T^d))$, cf.\ Eq.\ \eqref{IsoPhi} with $\ell^2(\ZZ^d)\cong L^2(\TT^d)$, explicitly defined by
$$ \hat{f}(\chi,x)=\sum_{\gamma\in\ZZ^d}f(x+\gamma)\chi(\gamma)^{-1}\equiv \sum_{\gamma\in\ZZ^d}(\gamma^*f)(x)\chi(\gamma)^{-1},\quad \chi\in\TT^d, x\in T^d.$$
This sum of $\gamma$-shifted versions of $f$ weighted by the phase factor $\chi(\gamma)^{-1}$ is often called a \emph{Bloch sum}. If we replace $x$ in the Bloch sum by $x+\gamma$, $\gamma\in\ZZ^d$, we get the \emph{Bloch wave} condition 
\begin{equation*}
\hat{f}(\chi,x+\gamma)=\hat{f}(\chi,x)\chi(\gamma). 
\end{equation*}
Thus each fixed \emph{quasimomentum} $\chi\in\TT^d$, the function $\hat{f}(\chi,\cdot)\in L^2(T^d)$ extends to a $\chi$-quasiperiodic \emph{Bloch ``wavefunction''} on $\RR^d$, albeit not normalisable over $\RR^d$ but only over $T^d$. Equivalently, we write $\hat{f}(\chi,\cdot)\in L^2(T^d;\cL_\chi)$ where each $\cL_\chi\rightarrow T^d$ is the line bundle obtained by quotienting $\RR^d\times \CC$ by $(x,z)\sim(x+\gamma,\chi(\gamma)z)$, $\gamma\in\ZZ^d$. Thus in total (see \S D.3 of \cite{FM}), there is a Hilbert bundle $\cE\rightarrow\TT^d$ with fibres 
\begin{equation*}
\cE_\chi=L^2(T^d;\cL_\chi),
\end{equation*}
 and $\hat{f}\in L^2(\TT^d;\cE)$ is an $L^2$-section of $\cE$. The action of translation by $\gamma\in\ZZ^d$ is represented unitarily on these sections by pointwise multiplication by the continuous function $\chi\mapsto\chi(\gamma)^{-1}$. 
An inversion formula
\begin{equation}
f(x)=\int_{\TT^d} \hat{f}(\chi,x)\,{\rm d}\chi\label{eqn:inverseBF}
\end{equation}
holds, recovering $f\in L^2(\RR^d)$ as a ``superposition'' of Bloch wavefunctions $\hat{f}(\chi,\cdot)\in L^2(T^d;\cL_\chi)$. In reverse, given a section $\phi\in L^2(\TT^d;\cE)$, its \emph{Wannier function} $w\in L^2(\RR^d)$ is the inverse Bloch--Floquet transform \eqref{eqn:inverseBF} of $\phi$.

Often, a certain finite-rank subbundle $\cE_S$ of $\cE$ is of interest, e.g.\ if the spectrum of $D$ has band structure, the spectral subspace $L^2_S(\RR^d)$ for spectra lying between some given spectral gaps is a $\Gamma$-invariant subspace obtainable (after taking Bloch--Floquet transform) as the $L^2$-sections of a locally trivial subbundle $\cE_S$, called a \emph{Bloch bundle} \cite{Panati, Kuchment, FM}. If there are $n$ bands (so $\cE_S$ has rank $n$), one can always find orthonormal measurable sections $\phi_j, j=1,\ldots, n$ for $\cE_S$; then each $\phi_j$ gives rise to a corresponding Wannier wavefunction $w_j\in L^2(\RR^d)$ such that the translates $\gamma^*w_j$ are mutually orthonormal \cite{Kuchment}. So we obtain an orthonormal \emph{Wannier basis} $\gamma^*w_j$, $j=1,\ldots, n$, $\gamma\in\ZZ^d$ for the spectral subspace $L^2_S(\RR^d)$ of interest.

\begin{figure}[ht]
\includegraphics[width=\textwidth]{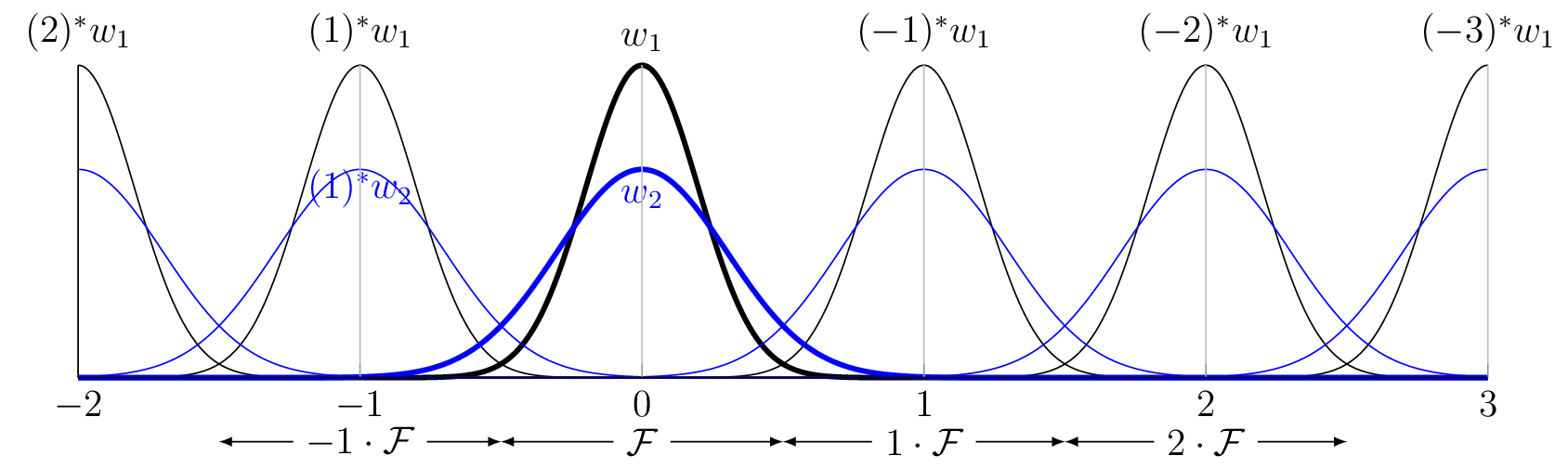}
\caption{Schematic depiction of a good Wannier basis (phase information not indicated) $\gamma^*w_j$, $j=1, 2$, $\gamma \in \ZZ$, obtained from an identification of $P_S=L^2_S(\RR)\cap L^2_\ZZ(\RR)$ as a free $C^*_r(\ZZ)=C(\TT)$-module of rank 2.}\label{fig:Zgoodwannier} \end{figure}

As mentioned in the introduction, when $d=2$, the Chern class of $\cE_S$ obstructs choosing the $\phi_j$ to be continuous, thereby obstructing
the existence of a Wannier basis comprising {\em exponentially decaying} wavefunctions \cite{Brouder}, or even much more mildly decaying ones \cite{Kuchment}. This failure has been turned into a triumph in recent years, because of the experimental discovery and burgeoning theoretical interest in these \emph{topological insulators} in physics, usually characterised exactly by the topological nontriviality of $\cE_S$ as detected, e.g.\ by $K$-theory classes \cite{Kitaev,FM,Thiang}. 

The same obstruction can occur for a range of decay conditions on the Wannier basis. From a physical perspective, a working definition of a topological insulator is one for which $L^2_S(\RR^d)$ does not admit an ``atomic limit'' \cite{Bradlyn}, which we can think of as the nonexistence of localised Wannier bases for $L^2_S(\RR^d)$. For this, a weak decay condition is preferred, so as to argue that a topological insulator necessarily has very delocalised Wannier bases. On the other hand, when the obstruction is not present, we would like to be able to choose Wannier wavefunctions which are as regular (smooth) and/or localised as possible. Let us remark that the rather extreme condition of \emph{compactly supported} Wannier wavefunctions was considered in \cite{Read} in connection with algebraic $K$-theory obstructions, cf.\ our construction of $L^2_c(X,E)$ as a pre-Hilbert $\CC[\Gamma]$-module in \S\ref{SectionExamples}.

\emph{For our purposes, we say that a (Wannier) wavefunction is ``good'' if it belongs to the Schwartz class $\cS(\RR^d)$, cf.\ Definition \ref{defn:Schwartzclass}.}

\subsection{Existence of good Wannier basis: nonabelian symmetry groups}

Quite generally, we can ask the nonabelian-$\Gamma$ analogue of the good Wannier basis existence problem: 

\begin{problem}\label{prob:good.Wannier}
Given a spectral subspace $L^2_S(X,E)$ of an admissible $\Gamma$-invariant operator $D$ on $L^2(X,E)$ (as defined in \S\ref{sec:submodules.section}-\ref{main.theorem.section}), does there exist a good Wannier basis, i.e.\ a set $w_1, \dots, w_n \in \cS(X,E)$ which together with their translates $\gamma^* w_j, \gamma\in\Gamma$ , form an orthonormal basis for $L^2_S(X, E)$?
\end{problem}

When $\Gamma$ has polynomial growth, our Thm.~\ref{TheoremMainTheorem} says that the dense subspace $P_S=L^2_S(X,E)\cap L^2_\Gamma(X,E)$ of ``not-so-poorly decaying'' functions forms a f.g.p.\ Hilbert $C^*_r(\Gamma)$-module. Furthermore, if $P_S$ is freely generated by $w_1,\ldots, w_n$, we can even choose these $w_j$ to be ``good'', i.e.\ in $\cS(X,E)$. Thus, Thm.~\ref{TheoremMainTheorem} answers Problem \ref{prob:good.Wannier} in the affirmative, for arbitrary $L^2_S(X,E)$, if all f.g.p.\ modules over $C^*_r(\Gamma)$ are free. A simple example where this occurs is $\Gamma=\ZZ$ (see Fig.~\ref{fig:Zgoodwannier}), and a new wallpaper group example is given in the next subsection (see Fig.~\ref{fig:pgunitcell}). Generically, there can be f.g.p.\ $C^*_r(\Gamma)$-modules which are not free, in which case our Thm.~\ref{TheoremMainTheorem} answers Problem \ref{prob:good.Wannier} in the negative ---  we can only achieve a {\em tight frame} if we want $w_j$ to be in $P_S$ or better.

The semigroup structure of f.g.p.\ $C^*_r(\Gamma)$-modules is difficult to understand in general, so a tractable first step is to compute $K_0(C^*_r(\Gamma))$. For crystallographic groups, one can appeal to the Baum--Connes conjecture, or use algebraic topology methods after converting $K_0(C^*_r(\Gamma))$ to a twisted equivariant $K$-theory group of $\TT^d$ as in \cite{FM}. For physical application, a reasonable justification for restricting the search to stably (non-)free f.g.p.\ modules can be made, following \cite{Kitaev}: $L^2_S(X,E)$ may be supplemented by well-localised inner atomic shells (free $C^*_r(\Gamma)$-modules) which were not accounted for in the specification of $D$. Thus the $K$-theory of $C^*_r(\Gamma)$ provides physically meaningful invariants which label \emph{topological phases} given a group $\Gamma$ of symmetries \cite{Thiang}.

\begin{figure}[ht]
\includegraphics[width=\textwidth]{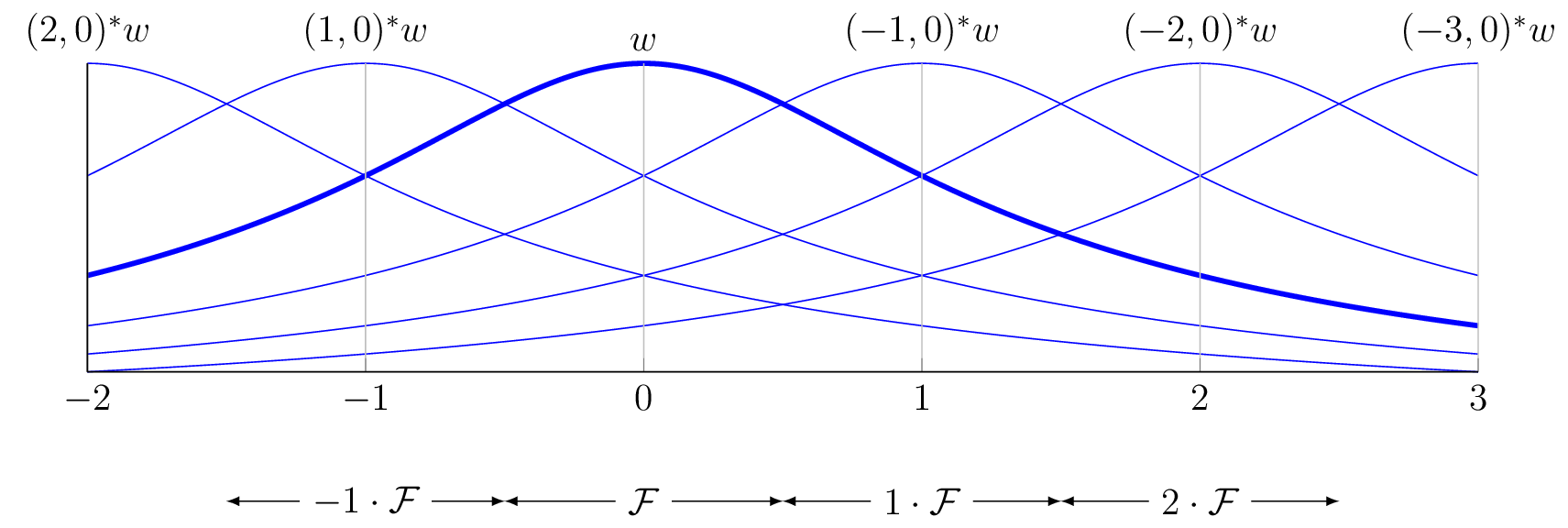}
\caption{If $P_S=L^2_S(\RR^2)\cap L^2_{\ZZ^2}(\RR^2)$ is a projective $C^*_r(\ZZ^2)=C(\TT^2)$-module which is not free, any Wannier basis for $L^2_S(\RR^2)$ must comprise wavefunctions with poor decay properties, as illustrated above with one spatial dimension suppressed.}\label{fig:Z2badwannier}
\end{figure}

\subsection{Crystallographic group examples}\label{sec:nonabelian.examples}
In this subsection, $E$ will be a trivial bundle over Euclidean $\RR^d$ and we omit reference to it. The $\Gamma$-invariant Hamiltonians whose spectral subspaces we are considering are assumed to satisfy the generic conditions of Thm.~\ref{ThmBlackBox}, so the conclusions of the main theorem \ref{TheoremMainTheorem} apply.

\medskip

\noindent {\bf Case $\Gamma=\ZZ^d$.}
The crucial difference between $\Gamma=\ZZ$ and $\Gamma=\ZZ^2$ is that all f.g.p.\ modules over $C^*_r(\ZZ)$ are free, whereas there are topologically nontrivial bundles over $\TT^2$, thus non-free f.g.p.\ modules over $C^*_r(\ZZ^2)$. A vector bundle over $\TT^2$ (f.g.p.\ $C^*_r(\ZZ^2)$-module) is trivial iff its first Chern class $c_1$ vanishes. Furthermore, the reduced $K$-theory of $\TT^2$ coincides with $c_1$ in this case, so a nonvanishing reduced $K$-theory class detects not just the failure of a f.g.p. $C^*_r(\ZZ^2)$-module to be stably free, but its failure to be free. 

The first Chern class obstruction, $c_1$, leads to ``Chern topological insulators'' in $d=2$, for which there is no good Wannier basis for $L^2_S(\RR^2)$. Note that we can measurably trivialise a topologically non-trivial Bloch bundle, and obtain a \emph{bad} Wannier basis for $L^2_S(\RR^2)$, e.g.\ \S3 of \cite{Kuchment}, as illustrated in Fig.\ \ref{fig:Z2badwannier}.
For $d>2$, there are further $K$-theory/Chern class obstructions to being free/stably free.

\medskip

\noindent {\bf Case $\Gamma=\ZZ\rtimes\ZZ$.}
An interesting non-abelian example is $\Gamma=\ZZ\rtimes\ZZ$, where the second copy of $\ZZ$ acts on the first by the nontrivial automorphism of reflection. In crystallography, this group appears as the 2D crystallographic space group (a ``wallpaper group'') $\pg$, and can be realised as a group of isometries of 2D Euclidean space $X=\RR^2$. Since $\Gamma$ is torsion free, the action is free and the fundamental domain is a manifold diffeomorphic to the Klein bottle, as illustrated in Fig.\ \ref{fig:pgunitcell}.

\begin{figure}[h]
\begin{center}
\subfigure{
\includegraphics[width=.55\textwidth]{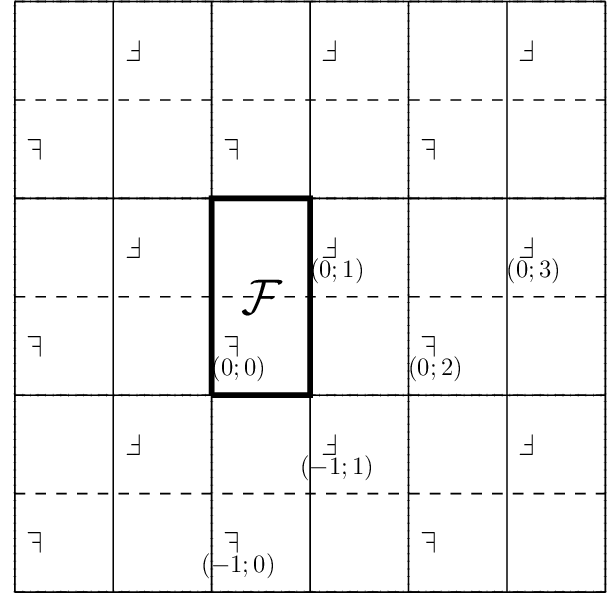}
}
\hspace{1.5em}
\subfigure{
\includegraphics[width=.25\textwidth]{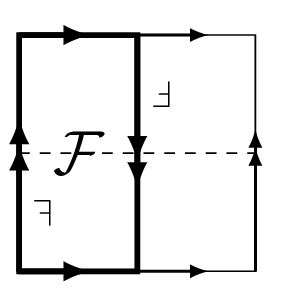}
}
\end{center}
\caption{(L) An embedding of $\pg$ in Euclidean space. The second copy of $\ZZ$ in $\pg=\ZZ\rtimes\ZZ$ acts by a \emph{glide reflection}, i.e.\ reflection in a horizontal glide axis (dotted line) followed by half-translation. Black solid lines enclose a choice of fundamental domain $\cF$, while the translates of $\cF$ are labelled uniquely by $(n_1;n_2)\in\pg$. (R) Illustration of $\cF$ as a Klein bottle, with opposite edges identified according to the arrows. A good Wannier basis always exists for a spectral subspace $L^2_S(\RR^2)$ of a generic $\pg$-invariant Hamiltonian, as in Corollary \ref{cor:pg.good.Wannier.basis}. We can visualise such a basis as smooth, localised ``atomic orbitals'' centered at the sites labelled by the $\Finv$.}\label{fig:pgunitcell}
\end{figure}

\begin{proposition}
Any finitely generated projective module over $C^*_r(\pg)$ is isomorphic to the free module $C^*_r(\pg)^n$ for some unique $n$.
\end{proposition}
\begin{proof}
The Klein bottle is a $B\pg$, and invoking the Baum--Connes isomorphism \cite{BCH} and low-dimension of $B\pg$, we compute that 
$$\ZZ\cong H_0(B\pg)=H_{\rm even}(B\pg)\cong K_0(B\pg)\cong K_0(C^*_r(\pg)).$$
The isomorphism takes the generator to the class $[1]$ of the identity projection in $C^*_r(\pg)$, or equivalently, the free rank-1 module, cf.\ \S7 of \cite{Valette} or \S 5.3 of \cite{GT1}. Thus every f.g.p.\ $C^*_r(\pg)$-module is stably isomorphic to $C^*_r(\pg)^n$ with $n\cdot[1]$ its $K$-theory class. Now, we note that $C^*(\ZZ)$ has a canonical faithful finite trace which extracts the coefficient at the identity, and that this trace is invariant under the reflection action of the second $\ZZ$ in $\pg\cong\ZZ\rtimes\ZZ$. Furthermore, $C^*(\ZZ)\cong\TT$ has (topological) stable rank 1 in the sense that invertibles are dense in $C^*(\ZZ)$, cf.~Prop.~1.7 of \cite{Rieffel}. So Rieffel's Thm.~10.8 in \cite{Rieffel} applies, saying that every stably free f.g.p.\ $C^*_r(\pg)$-module is actually already free.
\end{proof}

This nice property of $\pg$ allows us to invoke Thm.~\ref{MainTheoremIntro}, to answer the good Wannier basis existence problem in the affirmative:
\begin{corollary}\label{cor:pg.good.Wannier.basis}
Let the crystallographic group $\Gamma=\pg=\ZZ\rtimes\ZZ$ act on the Euclidean plane $X=\RR^2$ as above, and let $D$ be a Hamiltonian satisfying the conditions of Thm.~\ref{MainTheoremIntro}. Then for compact separated part $S \subset \mathrm{Spec}(D)$, the module $P_S$ from Thm.~\ref{TheoremMainTheorem} is a \emph{free} $C^*_r(\pg)$-module (of rank $n$ say), so that there exists a good Wannier basis $\gamma^*w_j$, $j=1,\ldots,n$, $\gamma\in\pg$, for $L^2_S(\RR^2)$.
\end{corollary}
To our knowledge, our technique is the first one that can achieve a result such as Corollary \ref{cor:pg.good.Wannier.basis} for nonabelian $\Gamma$.
\medskip

\begin{figure}[ht]
\includegraphics[width=\textwidth]{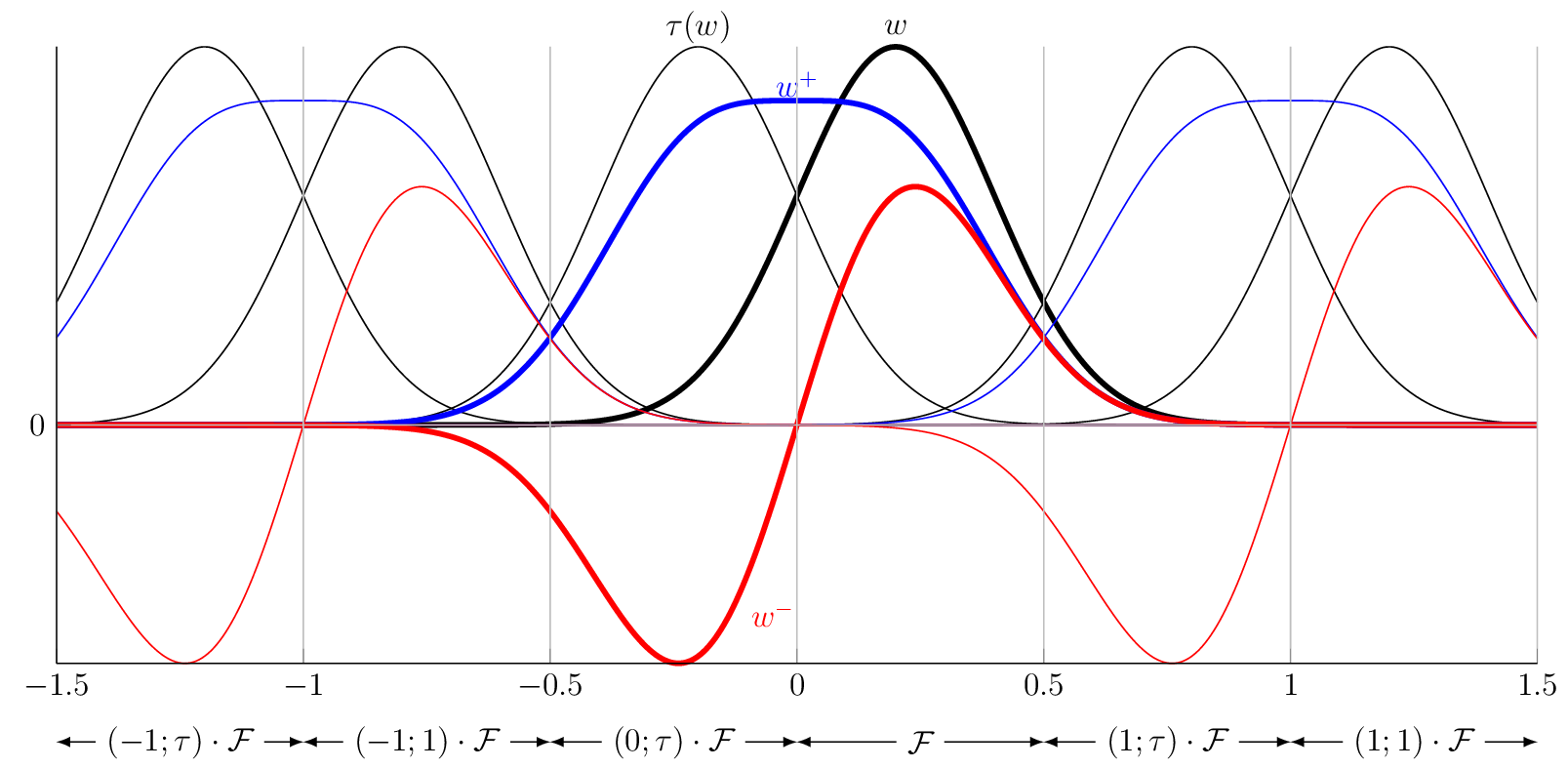}
\caption{If the module $P_S$ corresponding to  a spectral subspace $S \subset \mathrm{Spec(D)}$ is a free $C^*_r(\ZZ\rtimes\ZZ_2)$-module, a good Wannier basis comprising $w$ and its mutually orthogonal translates by $\gamma\in\ZZ\rtimes\ZZ_2$ can be constructed (black). This free module may be written as $P_S^+\oplus P_S^-$ where the submodule $P_S^+$ (resp.\ $P_S^-$) is generated (non-freely) by the thick blue (resp.\ thick red) Wannier function $w^+:=\frac{1}{\sqrt{2}}(w+\tau(w))$ (resp.\ $w^-:=\frac{1}{\sqrt{2}}(w-\tau(w))$). Notice that $\tau(w^+)$ and $w^+=\tau(w^+)$ are \emph{not} orthogonal, similarly for $\tau(w^-)$ and $w^-=-\tau(w^-)$. Decompose $L^2_S(\RR)=L^2_{S^+}(\RR)\oplus L^2_{S-}(\RR)$ where the $\ZZ\rtimes\ZZ_2$-invariant subspace $L^2_{S^+}(\RR)$ is spanned by the $\ZZ$-translates of $w^+$ and similarly for $L^2_{S^-}(\RR)$; their respective f.g.p.\ $C^*_r(\ZZ\rtimes\ZZ_2)$-modules are $P_S^+$ and $P_S^-$. A tight frame for $L^2_{S^+}(\RR)$ comprises $\frac{1}{\sqrt{2}}w^+$, $\frac{1}{\sqrt{2}}\tau(w^+)$ and their translates by $n\in\ZZ$; this tight frame has repeated elements since $w^+=\tau(w^+)$, so it is not an orthonormal basis. The case is  similar for $L^2_{S^-}$.}\label{fig:reflectionWannier}
\end{figure}

\noindent {\bf Case $\Gamma=\ZZ\rtimes\ZZ_2$.}
Maybe the simplest non-abelian crystallographic group is $\ZZ\rtimes\ZZ_2$, with $\ZZ_2$ acting on $\ZZ$ as $n\mapsto -n$. The group $\ZZ\rtimes\ZZ_2$ acts on $\RR$ with the (lift of the) generator $\tau$ of $\ZZ_2$ effecting reflection about some origin. We may compute that
$$K_0(C^*_r(\ZZ\rtimes\ZZ_2))\cong R(\ZZ_2)\oplus\ZZ M^\pm\cong\ZZ^2\oplus\ZZ=\ZZ^3,$$
where $R(\ZZ_2)\cong \ZZ^2$ denotes the representation ring of $\ZZ_2$, and $M^\pm$ is an extra projective $C^*_r(\ZZ\rtimes\ZZ_2)$-module which we will construct later (Fig.\ \ref{fig:reflectionWannierweird}). 
This computation should be contrasted with
$$K_0(C^*_r(\ZZ\times\ZZ_2))\cong K_0(C^*_r(\ZZ))\otimes K_0(C^*_r(\ZZ_2))\cong \ZZ\otimes_\ZZ R(\ZZ_2)\cong\ZZ^2.$$
The $R(\ZZ_2)\subset K_0(C^*_r(\ZZ\rtimes\ZZ_2))$ part is generated by f.g.p.\ $C^*_r(\ZZ\rtimes\ZZ_2)$-modules induced from irreducible $\ZZ_2$-representations as follows. There is a natural action $\alpha$ of $\ZZ_2$ on $C^*_r(\ZZ)$. Let $V$ be a finite-dimensional representation of $\ZZ_2$ and endow $\cE=V\otimes C^*_r(\ZZ)$ with the diagonal $\ZZ_2$ action. Then $\cE$ is a ``free $(C^*_r(\ZZ),\ZZ_2,\alpha)$-module'' in the equivariant $K$-theory sense described in \S11.2 of \cite{Blackadar}, and f.g.p.\ $(C^*_r(\ZZ),\ZZ_2,\alpha)$-modules are direct summands of such $\cE$.  There is a natural way to turn f.g.p.\ $(C^*_r(\ZZ),\ZZ_2,\alpha))$-modules into f.g.p.\ $C^*_r(\Gamma)\rtimes \ZZ_2\cong C^*_r(\ZZ\rtimes\ZZ_2)$-modules, giving a correspondence of the equivariant $K_0^{\ZZ_2}(C^*_r(\ZZ))$ and $K_0(C^*_r(\ZZ\rtimes\ZZ_2))$, see \S11.7 of \cite{Blackadar}. If we let $V^{\rm reg}$ be the regular representation of $\ZZ_2$ (which contains one copy each of the trivial and sign representations), then the above construction recovers $V^{\rm reg}\otimes  C^*_r(\ZZ)$ as the basic free $C^*_r(\ZZ\rtimes\ZZ_2)$-module, and exhibits it as a direct sum of $V^{\rm triv}\otimes  C^*_r(\ZZ)$ and $V^{\rm sign}\otimes  C^*_r(\ZZ)$, see Fig.\ \ref{fig:reflectionWannier}. There is another summand $M^\pm$ of $V^{\rm reg}\otimes  C^*_r(\ZZ)$, illustrated in Fig.\ \ref{fig:reflectionWannierweird}.

Alternatively, in the Fourier transformed picture, $K_0(C^*_r(\ZZ\rtimes\ZZ_2))\cong K_0^{\ZZ_2}(C^*_r(\ZZ))\cong K^0_{\ZZ_2}(\TT)$ where $\TT$ has the dual ``flip'' action of character-conjugation (momentum-reversal in physics). Then one computes $K^0_{\ZZ_2}(\TT)\cong \ZZ^3$, with generators explicitly realised by equivariant line bundles, cf.\ Lemma 4.3 of \cite{Gomi}, which are physically the Bloch bundles corresponding to  $V^{\rm triv}\otimes  C^*_r(\ZZ), V^{\rm sign}\otimes  C^*_r(\ZZ)$, and $M^\pm$.

\begin{figure}[ht]
\includegraphics[width=\textwidth]{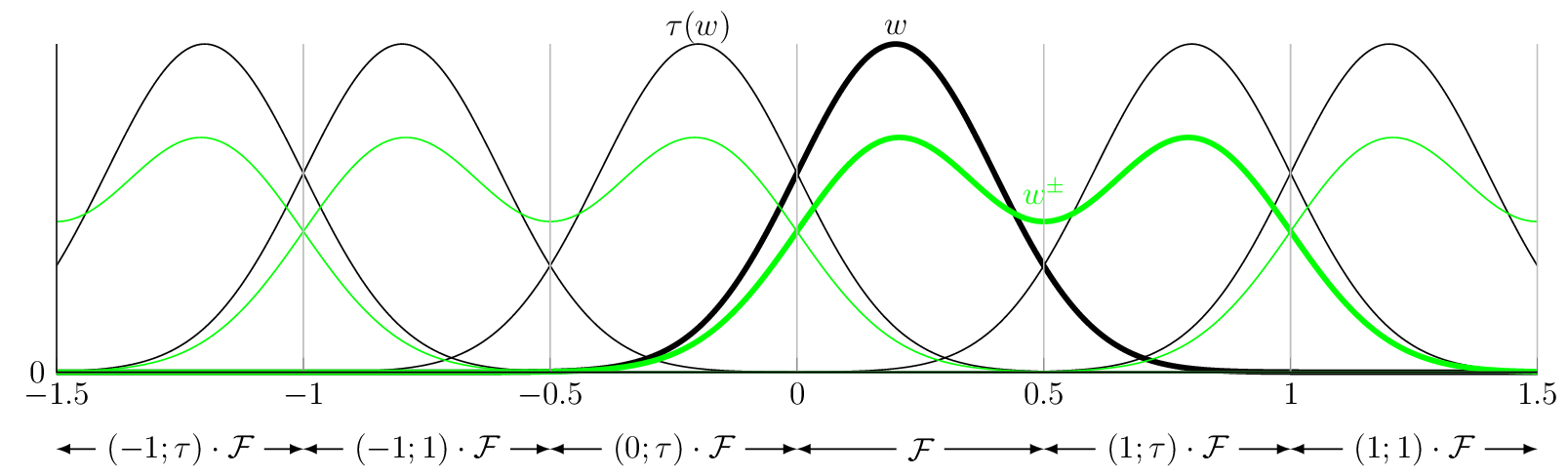}
\caption{Inside the rank-1 free $C^*_r(\ZZ\rtimes\ZZ_2)$-module $P_S$ of Fig.\ \ref{fig:reflectionWannier}, a submodule $M^\pm$ can be generated by $w^\pm:=\frac{1}{\sqrt{2}}(w+(1;\tau)^*w)$ (dark green). Note that $w^\pm$ has reflection centre at $0.5$, whereas $w^+$ of Fig.\ \ref{fig:reflectionWannier} has reflection centre at $0$.}\label{fig:reflectionWannierweird}
\end{figure}

\subsection{Magnetic translations and twisted group algebras}\label{sec:twisted.algebras}
Let $X$ be contractible (vanishing first and second cohomology also suffices), and let $\mathbf{B}$ be a $\Gamma$-invariant closed 2-form on $X$. For ${\rm dim}(X)=2$, $\mathbf{B}$ might arise as the curvature form of a magnetic field perpendicular to a surface. Pick a 1-form $\mathbf{A}$ (connection/vector potential) such that $d\mathbf{A}=\mathbf{B}$, and an origin $x_0\in X$.
For $\gamma \in \Gamma$,
let $\psi_\gamma $ be a function on $X$
satisfying $\gamma^*\mathbf{A}-\mathbf{A} = d \psi_\gamma$, such that
$\psi_\gamma(x_0)=0$ for all $\gamma  \in \Gamma$. For instance, take $\psi_\gamma(x)=\int_{x_0}^x (\gamma^*\mathbf{A}-\mathbf{A})$.
Let $S_\gamma$ denote multiplication by the phase function $\exp(i \psi_\gamma(x))$. Then the (left) \emph{magnetic translations} $T_\gamma^\sigma=(\gamma^{-1})^*\circ S_\gamma$
satisfy $T^\sigma_e={\rm{Id}}, \,\, T^\sigma_{\gamma_1}
T^\sigma_{\gamma_2} = \sigma(\gamma_1, \gamma_2)T^\sigma_{\gamma_1\gamma_2}$,
where $\sigma:\Gamma\times\Gamma\rightarrow {\rm U}(1)$ satisfies the group 2-cocycle condition,
\begin{align*}
\sigma(\gamma,e)&=\sigma(e,\gamma)=1,\qquad\qquad\qquad\gamma\in\Gamma,\\
\sigma(\gamma_1,\gamma_2)\sigma(\gamma_1\gamma_2,\gamma_3)&=\sigma(\gamma_1,\gamma_2\gamma_3)\sigma(\gamma_2,\gamma_3),\qquad\gamma_1,\gamma_2,\gamma_3,\in\Gamma.
\end{align*}
Here, we may verify each $\sigma(\gamma_1,\gamma_2)$ is indeed independent of $x\in X$, and is just a unimodular number.

Thus these magnetic translations furnish a \emph{projective} unitary representation of $\Gamma$ on $L^2(X)$ with cocycle $\sigma$, and one can verify that they commute with $\nabla_\mathbf{A}=d-i\mathbf{A}$ and thus with the magnetic Laplacian
$\frac{1}{2}\nabla_\mathbf{A}^*\nabla_\mathbf{A}$. A $\Gamma$-periodic potential may also be added.

On $\ell^2(\Gamma)$, the left regular representation can be twisted by $\sigma$, by taking
$(L^\sigma_\gamma a)(\gamma^\prime)=a(\gamma^{-1}\gamma^\prime)\sigma(\gamma,\gamma^{-1}\gamma^\prime)$,
obtaining the left $\sigma$-regular representation. We also have the right $\overline{\sigma}$-regular representation,
$(R^{\overline{\sigma}}_\gamma a)(\gamma^\prime)=a(\gamma^\prime\gamma)\overline{\sigma}(\gamma^\prime,\gamma)$,
where $\overline{\sigma}$ denotes the dual cocycle to $\sigma$. Using the cocycle identity, one sees that these projective representations commute with each other. 
We can construct the twisted group algebra $\CC[\Gamma,\overline{\sigma}]$ of finite linear combinations of $L^{\overline{\sigma}}_\gamma$, as well as the norm completion $C^*_r(\Gamma,\overline{\sigma})$ in $\cL(\ell^2(\Gamma))$.
Example \ref{ex:Roe.module} may be modified accordingly to obtain, inside $L^2(X)$, a right Hilbert $C^*_r(\Gamma,\overline{\sigma})$-module $L^2_{\Gamma,\sigma}(X)$ isomorphic to $L^2(\mathcal{F})\otimes C^*_r(\Gamma,\overline{\sigma})$ (cf.\ Proposition \ref{PropIsoToStandardModule}) via the map $\Phi_\mathcal{F}$ of Eq.\ \eqref{IsoPhi}. Up to this point, the discussion in this Subsection can also be found in \S B of \cite{Gruber} with minor differences in conventions.

\vspace{1em}

The good Wannier basis existence problem (Problem \ref{prob:good.Wannier}) carries over to this $\sigma$-twisted setting. Namely, one considers spectral subspaces $L^2_S(X)$ of admissible Hamiltonians acting on $L^2(X)$ which are $(\Gamma,\sigma)$-invariant, while the notion of Wannier basis just needs a modification of ``translates'' by ``$\sigma$-translates.'' 
A Wannier basis is thus some number of copies of the left $\sigma$-regular representation (realised inside $L^2(X)$).

If $\Gamma$ has polynomial growth, we can also construct the Fr\'{e}chet algebras of rapidly decaying sequences $H^\infty(\Gamma,\overline{\sigma},A)$ as in \S \ref{SectionRapidDecay}, since the relevant estimates are unaffected by the ${\rm U}(1)$ scalings. In particular, Lemma \ref{LemmaPolGrowth} and Proposition \ref{SectionRapidDecay} carry over in an identical way. Then the abstract results of \S\ref{sec:fgp.pre.Hilbert} and \S\ref{sec:submodules.section} follow, as does the Main theorem \ref{TheoremMainTheorem}. Thus we deduce that the solution to Problem \ref{prob:good.Wannier} in the $\sigma$-twisted setting is solved by checking whether 
\begin{equation*}
P_S=L^2_S(X)\cap L^2_{\Gamma,\sigma}(X)
\end{equation*}
is a free f.g.p.\ Hilbert $C^*_r(\Gamma,\overline{\sigma})$-module.

Let us mention that for $X$ the Euclidean plane, $\Gamma$ a lattice $\ZZ^2$ of translations, and $\mathbf{B}$ a constant magnetic field, the algebra $C^*_r(\ZZ^2,\overline{\sigma})$ is just the famous noncommutative torus, and there exist  non-trivial projective modules over $C^*_r(\ZZ^2,\overline{\sigma})$ --- the detection of such modules via $K$-theory and noncommutative Chern characters is of great importance in the quantum Hall effect \cite{Bellissard} modelled typically by magnetic Schr\"{o}dinger operators. Furthermore, the Baum--Connes conjectures also admit twisted versions, facilitating the computation of $K_0(C^*_r(\Gamma,\overline{\sigma}))$, in particular whether any non stably-free projective modules exist or not. If a spectral subspace (e.g.\ of a quantum Hall Hamiltonian) has $P_S$ being such a non-free module, then it cannot admit a good Wannier basis. Our analysis corroborates the physics heuristic that magnetic fields are usually responsible for lack of localization of the Wannier basis functions.

\subsection{A non-Euclidean example}\label{sec:non-Euclidean.example}
Consider the Heisenberg Lie group manifold 
$${\rm Heis}_\RR=\left\{\begin{pmatrix} 1 & x & z \\ 0 & 1 & y \\ 0 & 0 & 1 \end{pmatrix}\,:\, x,y,z\in\RR\right\}.$$ Its Lie algebra elements, i.e.\ tangent vectors at the identity, extend to left invariant vector fields on ${\rm Heis}_\RR$. An inner product on the Lie algebra similary extends to a left-invariant Riemannian metric on ${\rm Heis}_\RR$. Thus ${\rm Heis}_\RR$ is topologically $\RR^3$, but geometrically very different. Physically, such a non-Euclidean geometry could model a uniform density of screw dislocations along the $z$-direction \cite{Kleinert,HMT}. By restricting to $x,y,z\in\ZZ$, one obtains the the discrete subgroup ${\rm Heis}_\ZZ$, which of course acts on ${\rm Heis}_\RR$ isometrically by translations. Note that there is a central extension
$$1\rightarrow\ZZ\rightarrow{\rm Heis}_\ZZ\rightarrow \ZZ^2\rightarrow 1,$$
and decomposing over the character space $\TT$ of the central subgroup furnishes $C^*_r({\rm Heis}_\ZZ)$ as a continuous field $\{A_\theta\}_{\theta\in\TT}$ of noncommutative tori. It is known that $K_0(C^*_r({\rm Heis}_\ZZ))\cong \ZZ^3$, see e.g.\ \cite{AP}, with one of the generators being the class of the trivial projection, and that the $K$-theory class of a general projection may be computed via parings with cyclic cocycles \cite{HMT}. So we may compute in principle whether a f.g.p.\ $C^*_r({\rm Heis}_\ZZ)$-module is stably free or not.

\begin{proposition}
If a f.g.p.\ module over $C^*_r({\rm Heis}_\ZZ)$ is stably free, it is even a free module.
\end{proposition}
\begin{proof}
We require, for a $C^*$-algebra $A$, the notions of general stable rank, gsr($A$), and connected stable rank, csr($A$), as defined in \cite{Rieffel}. They are related by $1\leq$ gsr($A$)$\leq$csr($A$), see pp.\ 328 of \cite{Rieffel}, and the computation csr($C^*_r({\rm Heis}_\ZZ)$)$=2$ was carried out in \cite{Sudo}, so we have gsr($C^*_r({\rm Heis}_\ZZ))=1$ or $2$. As in the proof of Thm.~10.8 of \cite{Rieffel}, to show that gsr($C^*_r({\rm Heis}_\ZZ))=1$, it suffices to observe that $C^*_r({\rm Heis}_\ZZ)$ has a faithful finite trace.  Then by Corollary 10.7 of \cite{Rieffel}, gsr($C^*_r({\rm Heis}_\ZZ)$)$=1$ implies the claim of this Proposition.
\end{proof}

Note that ${\rm Heis}_\ZZ$ is nilpotent, thus of polynomial growth, and acts freely and isometrically on ${\rm Heis}_\RR$ with quotient/fundamental domain a compact nilmanifold. Thus we may apply our Main Thm.~\ref{MainTheoremIntro} to conclude that when a spectral subspace $L^2_S({\rm Heis}_\RR)$ of a ${\rm Heis}_\ZZ$-invariant Hamiltonian $D$ (satisfying the generic conditions of that theorem) has $P_S$ with $K$-theory class $n\cdot[1]\in K_0(C^*_r({\rm Heis}_\ZZ))$, we can already conclude that a good Wannier basis $\gamma^*w_j$, $j=1,\ldots,n$, $\gamma\in\mathrm{Heis}_\ZZ$ for $L^2_S({\rm Heis}_\RR)$ exists.

\medskip
\textbf{Acknowledgements.} We would like to thank Vito Zenobi, Varghese Mathai, Giuseppe De Nittis, Gianluca Panati, and Domenico Monaco for helpful discussion. The first-named author was supported by Australian Research Council Discovery Project grant FL170100020, under Chief Investigator and Australian Laureate Fellow Mathai Varghese, and the second-named author by Australian Research Council Discovery Early Career Researcher Award grant DE170100149 and Discovery Projects DP200100729. 

\textbf{Data availability statement.}
Data sharing is not applicable to this article as no new data were created or analyzed in this study.


\begin{thebibliography}{9}
\bibliographystyle{plain}
\bibitem{AP}
J. Anderson, W. Paschke.: The rotation algebra. Houston J. Math. {\bf 15}(1) 1--26 (1989)

\bibitem{BGP}
C. B\"ar, N. Ginoux, F. Pf\"aeffle.: Wave Equations on Lorentzian Manifolds and Quantization (ESI Lectures in Mathematics and Physics), EMS Publishing House, 2007.

\bibitem{BCH}
P. Baum, A. Connes, N. Higson.: Classifying space for proper actions and $K$-theory of group $C^*$-algebras. $C^*$-algebras: 1943--1993 (San Antonio, TX, 1993), Contemp. Math. {\bf 167}, Amer. Math. Soc. (1994), 240--291.

\bibitem{Bellissard}
J. Bellissard.: $K$-theory of $C^*$-algebras in solid state physics. Statistical mechanics and field theory: mathematical aspects, Springer, pp.\ 99--156, 1986

\bibitem{Blackadar}
B. Blackadar.: $K$-theory for operator algebras. Math. Sci. Res. Inst. Publ. vol. 5, Cambridge University Press, 1998.

\bibitem{BCR}
C. Bourne, A. Carey, A. Rennie.: The bulk-edge correspondence for the quantum Hall effect in Kasparov theory. Lett. Math. Phys. {\bf 105}(9) 1253--1273 (2015)

\bibitem{Bradlyn}
B. Bradlyn, L. Elcoro, J. Cano, M.G. Vergniory, Z. Wang, C. Felser, M.I. Aroyo, B.A. Bernevig: Topological quantum chemistry. Nature {\bf 547}(7663) 298 (2017)

\bibitem{Brouder}
C. Brouder, G. Panati, M. Calandra, C. Mourougane, N. Marzari.: Exponential Localization of Wannier Functions in Insulators. Phys. Rev. Lett. {\bf 98} 046402 (2007)

\bibitem{BrownOzawa}
N. P. Brown, N. Ozawa.: {$C^*$}-algebras and finite-dimensional approximations. Amer. Math. Soc., Graduate Studies in Mathematics vol. 88, 2008.

\bibitem{BS}
J. Br\"{u}ning, T. Sunada.: On the spectrum of periodic elliptic operators. Nagoya Math. J. {\bf 126} 159--171 (1992)

\bibitem{CGT}
J. Cheeger, M. Gromov, M. Taylor.: Finite propagation speed, kernel estimates for functions of the Laplace operator, and the geometry of complete Riemannian manifolds. J. Diff. Geom. {\bf 17}(1) 15--53 (1982)

\bibitem{DesCloi}
J. des Cloizeaux.: Energy Bands and Projection Operators in a Crystal: Analytic and Asymptotic Properties. Phys. Rev. {\bf 135}(3A) 685--697 (1964)

\bibitem{Cornean}
H.D. Cornean, I. Herbst, G. Nenciu.: On the Construction of Composite Wannier Functions. Ann. Henri Poincar\'{e} {\bf 17} 3361--3398 (2016)

\bibitem{CMM}
H.D. Cornean, D. Monaco, M. Moscolari.: Parseval frames of exponentially localized magnetic Wannier functions. Commun. Math. Phys. {\bf 371} 1179--1230 (2019)

\bibitem{DNPanati}
G. De Nittis, G. Panati.: The topological Bloch--Floquet transform and some applications. In Operator Theory: Advances and Applications, Vol. 224, 67--105, Springer, 2012

\bibitem{GGT}
C. Dru\c{t}u, M. Kapovich.: Geometric Group Theory. Amer. Math. Soc. Colloq. Publ., vol. 63, 2018

\bibitem{Engel}
A. Engel: Rough index theory on spaces of polynomial growth and contractibility. J. Noncommutative. Geom. {\bf 13}(2) 617--666 (2019)

\bibitem{EwertMeyer}
E. Ewert, R. Meyer.: Coarse geometry and topological phases. Commun. Math. Phys. {\bf 366}(3) 1069--1098 (2019)

\bibitem{FM}
D.S. Freed, G. Moore.: Twisted equivariant matter. Ann. Henri Poincar\'{e} {\bf 14}(8) 1927--2023 (2013)

\bibitem{Gomi}
K. Gomi.: A Variant of $K$-theory and Topological T-duality for Real Circle Bundles. Commun. Math. Phys. {\bf 334} 923--975 (2015)

\bibitem{GT1}
K. Gomi, G.C. Thiang.: Crystallographic bulk-edge correspondence: glide reflections and twisted mod 2 indices. Lett. Math. Phys. {\bf 109}(4) 857--904 (2019)

\bibitem{Gromov}
M. Gromov, with an appendix by Jacques Tits.: Groups of polynomial growth and expanding maps. Publ. Math. Inst. Hautes \'{E}tudes Sci. {\bf 53} 53--78 (1981)

\bibitem{Gruber}
M.J. Gruber.: Noncommutative Bloch theory. J. Math. Phys. {\bf 42}(6) 2438--2465 (2001)

\bibitem{Han}
D. Han, K. Kornelson, D. Larson, E. Weber.: Frames for undergraduates, Amer. Math. Soc., Student Math. Library vol. 40, 2007.

\bibitem{HMT}
K.C. Hannabuss, V. Mathai, G.C. Thiang.: T-duality simplifies bulk-boundary correspondence: the parametrised case. Adv. Theor. Math. Phys. {\bf 20}(5) 1193--1226 (2016)

\bibitem{Ji}
R. Ji.: Smooth Dense Subalgebras of Reduced Group $C^*$-Algebras, Schwartz Cohomology of Groups, and Cyclic Cohomology. J. Funct. Anal. {\bf 107}(1) 1--33 (1992)

\bibitem{Jolissaint}
P. Jolissaint.: Rapidly Decreasing Functions in Reduced $C^*$-Algebras of Groups. Trans. Amer. Math. Soc. {\bf 317}(1) 167--196 (1990)

\bibitem{Kellendonk}
J. Kellendonk.: On the $C^*$-Algebraic approach to Topological Phases for Insulators. Ann. Henri Poincar\'{e}, {\bf 18}(7) 2251--2300 (2017)

\bibitem{Kitaev}
A. Kitaev.: Periodic table for topological insulators and superconductors. AIP Conference Proceedings. {\bf 1134}(1) 22--30 (2009)

\bibitem{Kleinert}
H. Kleinert.: Gauge fields in condensed matter, vol II, World Scientific, 1989.

\bibitem{Kohn}
W. Kohn.: Analytic properties of Bloch waves and Wannier functions. Phys. Rev. {\bf 115} 809--821 (1959)

\bibitem{Kubota}
Y. Kubota.: Controlled Topological Phases and Bulk-edge Correspondence. Commun. Math. Phys. {\bf 349}(2) 493--525 (2017)

\bibitem{Kuchment}
P. Kuchment.: Tight frames of exponentially decaying Wannier functions. J. Phys. A: Math. Theor. {\bf 42} 025203 (2009)

\bibitem{Lance}
E.C. Lance.: Hilbert $C^*$-modules: a Toolkit for Operator Algebraists, London Math. Soc. Lect. Notes 210, Cambridge Univ. Press, 1995.

\bibitem{Ludewig-Thiang}
M. Ludewig, G.C. Thiang.: Cobordism invariance of topological edge-following states. arXiv:2001.08339

\bibitem{MM}
V. Mathai, M. Marcolli.: Twisted index theory on good orbifolds, I: Noncommutative Bloch theory. Commun. Contemp. Math. {\bf 1}(04) 553--587 (1999)

\bibitem{Monaco}
D. Monaco, G. Panati, A. Pisante, S. Teufel.: Optimal Decay of Wannier functions in Chern and Quantum Hall Insulators. Commun. Math. Phys. {\bf 359} 61--100 (2018)

\bibitem{Nenciu}
G. Nenciu.: Dynamics of band electrons in electric and magnetic fields: rigorous justification of the effective Hamiltonians. Rev. Mod. Phys. {\bf 63}(1) 91--127 (1991) 

\bibitem{Panati}
G. Panati.: Triviality of Bloch and Bloch--Dirac Bundles. Ann. Henri Poincar\'{e} {\bf 8} 995--1011 (2007)

\bibitem{PSB}
E. Prodan, H. Schulz-Baldes.: Bulk and Boundary Invariants for Complex Topological Insulators. Math. Phys. Studies, Springer, 2016.

\bibitem{Read}
N. Read.: Compactly supported Wannier functions and algebraic $K$-theory. Phys. Rev. B {\bf 95}(11) 115309 (2017)

\bibitem{Roe}
J. Roe.: Comparing analytic assembly maps. Quart. J. Math. {\bf 53} 241--248 (2002)

\bibitem{Rieffel}
M.A. Rieffel.: Dimension and stable rank in the $K$-theory of $C^*$-algebras. Proc. London Math. Soc. {\bf s3-46}(2) 301--333 (1983)

\bibitem{Sudo}
T. Sudo.: Stable rank of $C^*$-algebras of continuous fields. Tokyo J. Math. {\bf 28}(1) 173--188 (2015)

\bibitem{Taylor1}
M. Taylor.: Partial Differential Equations I: Basic Theory, Springer, Applied Mathematical Sciences, 2nd Edition, 2011.

\bibitem{Thiang}
G.C. Thiang.: On the $K$-theoretic classification of topological phases of matter. Ann. Henri Poincar\'{e} {\bf 17}(4) 757--794 (2016)

\bibitem{Valette}
A. Valette.: Introduction to the Baum--Connes conjecture. Lectures Math. ETH Z\"{u}rich, Birkh\"{a}user, 2002.

\bibitem{WO}
N.E. Wegge-Olsen.: $K$-theory and $C^*$-algebras. Oxford Univ. Press, 1993.

\end{thebibliography}
\end{document}